\documentclass[12pt]{article}
\usepackage[utf8]{inputenc}

\usepackage{latexsym,amssymb,amsfonts,amsmath,amsthm}
\usepackage[dvips]{graphicx}
\usepackage{comment}

\newtheorem{theorem}{Theorem}
\newtheorem{lemma}{Lemma}

\newtheorem{proposition}{Proposition}
\newtheorem{remark}{Remark}
\newtheorem{definition}{Definition}

\newtheorem{assumption}{Assumption}

\numberwithin{theorem}{section}
\numberwithin{lemma}{section}
\numberwithin{corollary}{section}
\numberwithin{proposition}{section}
\numberwithin{remark}{section}
\numberwithin{definition}{section}

\newcommand{\bs}[1]{\boldsymbol{#1}}

\newcommand{\spann}{{\rm span}}

\newcommand{\supp}{{\rm supp}}

\newcommand{\SL}[1]{\stackrel{\circ}{#1}}

\newcommand{\dist}{{\rm dist}}

\usepackage{color}

\title{Quantum walks on graphs embedded in orientable surfaces}
\author{Yusuke Higuchi$^1$, Etsuo Segawa$^2$\\
$^1${\small Department of Mathematics,
Gakushuin University,} \\
{\small Tokyo, 171-8588, Japan.}
\\
$^2${\small Graduate School of Environment and Information Sciences, Yokohama National University,}\\ {\small Hodogaya, Yokohama, 240-8501, Japan.
}
}
\date{}

\begin{document}

\maketitle

\par\noindent
{\bf Abstract}. 
 \\
A quantum walk model which reflects the $2$-cell embedding on the orientable closed surface of a graph in the dynamics is introduced. We show that the scattering matrix is obtained by finding the faces on the underlying surface which have the overlap to the boundary and the stationary state is obtained by  counting two classes of the rooted spanning subgraphs of the dual graph on the underlying embedding.  \\



\section{Introduction}
The convergence to the stationary state in the study on irreducible and ergodic random walks on finite graphs is very fundamental to explore a lot of interesting phenomena, for examples, cut off phenomenon~\cite{Diaconis,LevinPeres} and the relation to the electric circuit~\cite{DS}. Quantum walk has been introduced by a quantum analogue of random walks~\cite{Gudder, Meyer} and it exhibits several properties such as effectiveness on the quantum search algorithm on finite graphs~(see \cite{Ambainis2003,Childs,Portugal} and its references therein) and the coexistence of ballistic and localization~\cite{KonnoBook, HKSS} on the infinite lattices. Due to the unitarity of the time evolution of a quantum walk, every eigenvalue lives on the unit circle in the complex plain. Such a fact induces the quasi-periodicity to every quantum walk in general~\cite[Theorem~7.7]{Portugal}. Thus we need to device something to obtain the stationary state in the construction of the quantum walk. The quantum walk models on semi-infinite graphs which accomplish the convergence to a fixed point as a dynamical system in the long time limit have been introduced in \cite{FelHil1,FelHil2}. 
In this model, the total time evolution is driven by an infinite dimensional unitary operator on this semi-infinite system  while the restricted time iteration to the internal graph can be regarded as a finite dynamical system receiving the inflow and radiating the outflow to the outside at every time step. The outside corresponds to the semi-infinite paths whose rooted vertices are identified with the boundaries of the original graph. These semi-infinite paths are called the tails. In addition, the time evolution of the quantum walk on the tails is free, that is, every quantum walker from the outside moves inward and one from the inside moves outward.    
It is mathematically shown that by the balance between such inflow and outflow, this dynamics converges to a fixed point~\cite{FelHil1,FelHil2,HS}.  
In \cite{SchSmi, SevTan}, the relation between discrete-time quantum walk models to the stationary Shr{\"o}dinger equation on the metric graph, say the quantum graph~\cite{ExSe, Kuch, Albe}, are clarified. 
The quantum walk model introduced by \cite{FelHil1} has also deeply related to the quantum graph and has been studied from the view point of the scattering theory on the Schr{\"o}dinger equation~\cite{MMOS,Morioka,KHiguchi} and also considered the fermionic system~\cite{AndEtAL}.  
In this paper, we treat such a quantum walk model on the graph with tails receiving the inflow and radiating the outflow~\cite{FelHil1,FelHil2,HS}. 

The time evolution operator of a quantum walk is determined by the pair of the underlying graph $G=(V,E)$, whose degree is locally finite, and the sequence of local unitary matrices assigned at each vertex $\{C_u\}_{u\in V}$. 
Here $C_u$ is a $\deg(u)$-dimensional unitary operator, where $\deg(u)<\infty$ is the degree of the vertex $u\in V$. The unitary matrix assigned at vertex $u$ is called the coin matrix of $u$.  
One of the typical choices of $C_u$'s are $C_u=2/\mathrm{deg}(u)\;J_u-I_u$ for any $u\in V$, where $J_u$ and $I_u$ are the all $1$ matrix and the identity matrix with dimension $\deg(u)$; such a quantum walk model is called the Grover walk. The useful property for the construction of the Grover walk is that the dynamics is uniquely determined by only the degree of each vertex which is independent of the labeling of arcs because the Grover matrix commutes with any permutation matrix~\cite{KroBru}. 
Thus it is sometimes possible to extract some graph geometries induced by the Grover walk~\cite{HSS2}. 
On the other hand, of course, there are infinite many choices of the local unitary matrices $\{C_u\}$ other than the Grover matrix. We are interested in the choice of quantum coins so that the total time evolution reflects not only the graph adjacency relation but also the way of the drawing the graph, in particular, {\it its underlying closed surface} to draw the graph without any crossings of edges. To this end, we consider the graph with the {\it rotation}. The rotation $\rho$ is the set of cyclic permutations assigned at each vertex whose length is the degree of each vertex. It is well known that the rotation determines the configuration of the vertices on the underlying oriented surface for the $2$-cell embedding of the graph~\cite{GT,MT,NakamotoOzeki}. Every facial closed walk, which is a closed walk on the graph naturally determined by the rotation, gives a region (in other word, a face) in the underlying $2$-cell embedded surface. Thus by Euler's theorem, the genus of the underlying orientable surface induced by the rotation $\rho$ is described by
\begin{equation}\label{eq:genus} 
g= \frac{1}{2}(b_1(G)-r(G;\rho)+1), 
\end{equation}
where $r(G;\rho)$ is the number of facial walks and $b_1(G)=|E|-|V|+1$ is the Betti number. 
See Section~\ref{sect:rotation} for the definition of the facial walk.
Up to the choice of rotation $\rho$, the genus of the orientable surface can be controlled. The rotation giving the maximal number of the closed facial walks provides the minimal genus, $\gamma(G)$, of the underlying $2$-cell embedding. On the other hand, the rotation giving the minimal number of the closed facial walks provides the maximal genus, $\gamma_M(G)$, of the underlying $2$-cell embedding. There are nice books on surveying the embedding of the graphs; for example, \cite{GT,MT} and \cite{NakamotoOzeki}. 
The following is quite interesting and useful theorem for the genus of the underlying orientable closed surface which is known as Duke's interpolation theorem (1966) (see \cite{MT} and its reference therein and see also Figure~\ref{fig:embedding} for an example):
\begin{theorem}[Duke's interpolation theorem]
For any connected graph $G$, and for any natural number $g$ with 
$\gamma(G)\leq g \leq \gamma_M(G)$, 
there exists a $2$-cell embedding of $G$ in the orientable closed surface with the genus $g$. 
\end{theorem}
\begin{figure}[hbtp]
    \centering
    \includegraphics[keepaspectratio, width=160mm]{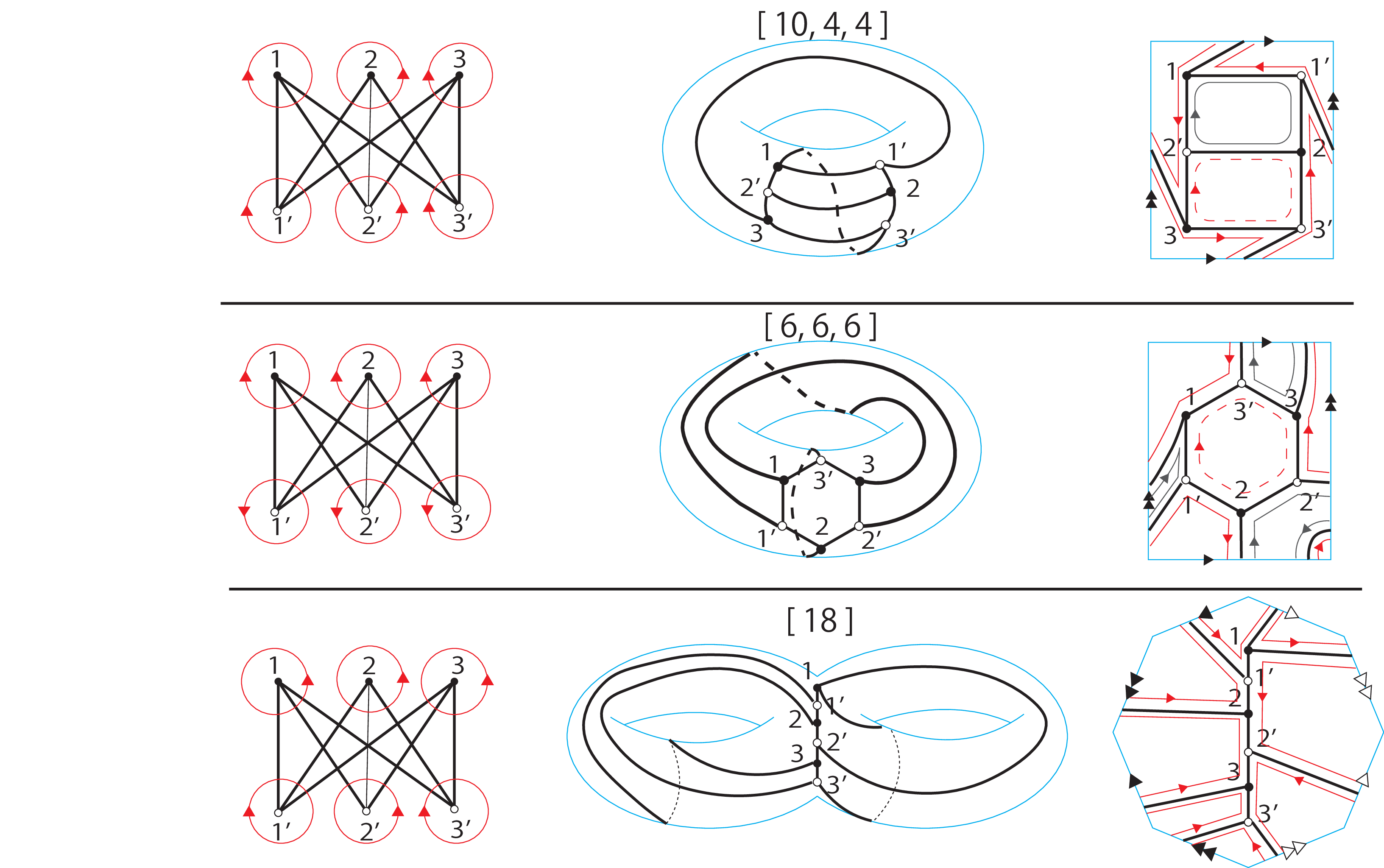}
    \caption{The underlying closed surface of $K_{3,3}$: It is known that the minimum genus of the underlying orientable closed surface for $K_{3,3,}$ is $\gamma(K_{3,3})=1$. On the other hand, since the minimum number of the facial walks is $1$ (see the right bottom of the figure), the maximum number of that is $\gamma_M(K_{3,3,})=2$ by (\ref{eq:genus}). Three kinds of rotations are depicted (the first column); we call them $[10,4,4]$-type, $[6,6,6]$-type and $[18]$-type, respectively. The orientable underlying surface and the $2$-cell embedding of $K_{3,3}$ are determined by the rotation (the second column): the genuses of the underlying surfaces induced by the first and second rotations, $[10,4,4]$ and $[6,6,6]$ types, are $1$, but their embedding ways are different because there are $2$-square faces and $1$-decagon faces in the first embedding while there are $3$ hexagon faces in the second embedding (the third column). On the other hand, the third rotation, $[18]$-type, needs the maximal genus $\gamma_M(K_{3,3})=2$ of the underlying orientable surface. 
    The way to assign the vertices on the closed surface is as follows: (i) on each $v\in V$, arrange radially the edges connected to $v$ following the rotation $\rho_v$; (ii) connect the vertices so that on each edge, the rotations of the two end vertices must be opposite direction~\cite{MT}. }
    \label{fig:embedding}
\end{figure}

Then in this paper, we treat a quantum walk model which can  reflect the underlying $2$-cell embedding in the orientable surface of the graph to the dynamics, named the {\it facial quantum walk}. 
Originally, the facial quantum walk model was called the {\it optical quantum walk} for the purpose of the design of an implementation of a quantum walk model which has a fixed point as a dynamical system by the optical polarizing elements~\cite{Mizutani_ETAL}. 
The polarization is represented by blowing up the original graph~\cite{Mizutani_ETAL}.  
In addition we notice such a potential of this quantum walk in the description of an underlying $2$-cell embedding of this quantum walk model from the above famous theorem.  As an interesting related work, the quantum walk with the permutated Grover coins following the edge coloring of the planner graph is constructed and the dark state is described by the expression of the underlying edge coloring~\cite{MaresNovotonyJex}. 

The dynamics of the facial quantum walk is determined by the abstract graph, the rotation which is determined by the $2$-cell embeding on the orientable surface, the place of the tails and the $2\times 2$ unitary matrix $H$. 
In this paper, we describe the  scattering matrix in the stationary state of the facial quantum walk on the blow up graph with tails (see Theorem~\ref{thm:scattering}). 
Interestingly, the scattering matrix can be obtained by only the information on the {\it external} facial walks which pass through the boundary vertices. To extract the internal graph geometry, we consider the stationary state restricted to the internal graph. 
Then we obtain the expression for the stationary state on $G^{BU}$ 
described by counting the typical spanning subgraphs of the dual graph $G^*$ induced by the underlying $2$-cell embedding $(G,\rho)$, and the boundary of the graph (see Theorem~\ref{thm:stationary}).  

This paper is organized as follows. 
In section~2, for a give graph $G$ with a rotation $\rho$, we give the definition of the blow up graph $G^{BU}$ which is the digraph both of whose in-degree and out-degree are $2$. The graph contains the information about the embedding of $G$ into some orientable surface and implements an optical polarizing circuit~\cite{Mizutani_ETAL}. Moreover we introduce the definition of quantum walk model on a blow up graph $G^{BU}$, the facial quantum walk on $G^{BU}$. 
In section~3, we show that the scattering matrix is expressed by using the facial walk having the overlap to the boundary (Theorem~\ref{thm:scattering}). 
We demonstrate the scattering matrices on the three kinds of $2$-cell embedding on orientable surface of $K_{3,3}$ as an example, and propose an idea to detect the underlying embedding way of the graph by just observing the scattering. 
In section~4, we show that the stationary state on $G^{BU}$ can be described by using the closed facial walks whose length fit the frequency of the determinant of the quantum coin.
We demonstrate that our quantum walk can be controlled to choose the faces of the truncated icosahedron which has the overlap to the stationary state so that it colors the soccer ball pattern by adjusting the determinant of the quantum coin. We also show that the stationary state on $G^{BU}$ can be described by the number of rooted spanning trees of the dual graph $G^*$ of $(G,\rho)$ which is induced by the embedding on the orientable closed surface (Theorem~\ref{thm:stationary}). 
Finally, we compute the stationary state on the tetrahedron as an example. 
\section{Setting}
\subsection{Blow up of rotation graph}\label{sect:rotation}
Let $G=(V,A)$ be a finite, connected, simple and symmetric digraph. Here the ``symmetric" means that  for any arc $a\in A$, there uniquely exists an inverse arc $\bar{a}$ in $A$. The boundary of $G$, $\delta V$, is a subset of $V$. 
We connect the semi-infinite paths, called the tails, to the boundary identifying the origin of each tail with each boundary vertex. 
Here the vertex and the arc sets of the tails are denoted by $V_{tl}$, $A_{tl}$, respectively. 
The set of the ``tips" of the tails are defined by $\delta A=\delta A^+ \cup \delta A^-$ with $\delta A^+=\{ e\in A_{tl} \;|\; t(e) \in \delta V \}$, $\delta A^-=\{ e\in A_{tl} \;|\; o(e) \in \delta V \}$, respectively. 
Let $\tilde{N}_u$ be the neighborhood of $u$ in the resulting infinite graph. 
Note that if $u\in \delta V$, then $|\tilde{N}_u|=d_G(u)+1$, where $d_G(u)$ is the degree of $u$ in the original graph $G$, since the boundary vertex has a neighbor of the tail.

We also introduce newly important notion as follows. 
We assign a cyclic permutation $\rho_u: \tilde{N}_u\to \tilde{N}_u$ with length $|\tilde{N}_u|$ to each vertex $u\in V$.  
It is natural to construct the map $\rho: A\cup \delta A^+\to A\cup \delta A^-$ induced by $\{\rho_u\}_{u\in V}$ as follows:  for any $e\in A\cup \delta A^+$, 
\[ o(\rho(e))=t(e) \text{ and } t(\rho(e))=\rho_{t(e)}(o(e)). \]
Then the triple of $\tilde{G}:=(G;\delta V;\rho)$ is called the rotation tailed graph. See Figure~\ref{fig:1}. 

In this paper, we focus on its deformed graph $G^{BU}$; we call it the {\it blow up graph} of $\tilde{G}$. 
The blow up graph is obtained by replacing each original vertex of the internal of $\tilde{G}$ with the oriented cycle induced by $\rho$. 
Let us explain it more precisely as follows. 
Let $V^{BU}$ and its subset $\delta V^{BU}$ be denoted by 
\begin{align*} 
& V^{BU} = \{(u,v)\;|\; {u\in V},\;{v\in \tilde{N}_u}\}, \\
& \delta V^{BU} = \{ (t(e),o(e)) \;|\; e\in\delta A^{+} \},
\end{align*}
respectively. The set of arcs $A^{BU}\subset V^{BU}\times V^{BU}$ is denoted by 
\[ e\in A^{BU}\text{ with $o(e)=(u,v),\;t(e)=(u',v')\in V^{BU}$} \Leftrightarrow\; 
      \begin{cases}
      \text{(i)} & u=u' \text{ and } v'=\rho_u(v) \\
      &\text{\qquad or} \\
     \text{(ii)} & u=v' \text{ and } v=u'.  
      \end{cases} \]
For a finite graph $(V^{BU},A^{BU})$, the blow up graph $G^{BU}$ is obtained by 
connecting a tail to each vertex of $\delta V^{BU}$, which is an infinite graph. See Figure~\ref{fig:1}.
Here let us use the same notation with that of $\tilde{G}$ to denote the set of arcs of such tails by $A_{tl}$ and let us define the special subset of $A_{tl}$ as ``piers" of the blow up graph $G^{BU}$ by  
\[\delta A_{pr}=\delta A_{pr}^{+}\cup \delta A^{-}_{pr}\]
with $\delta A^+_{pr}=\{ e\in A_{tl} \;|\; t(e) \in \delta V^{BU} \}$, $\delta A^-_{pr}=\{ e\in A_{tl} \;|\; o(e) \in \delta V^{BU} \}$ 
for the blow up graph $G^{BU}$. 
Then the vertex and arc sets of $G^{BU}$ is described by
\begin{align}
V(G^{BU}) &= V^{BU} \cup V_{tl}, \\
A(G^{BU}) &= A^{BU} \cup A_{tl}.
\end{align} 

%
\begin{figure}[hbtp]
    \centering
    \includegraphics[keepaspectratio, width=120mm]{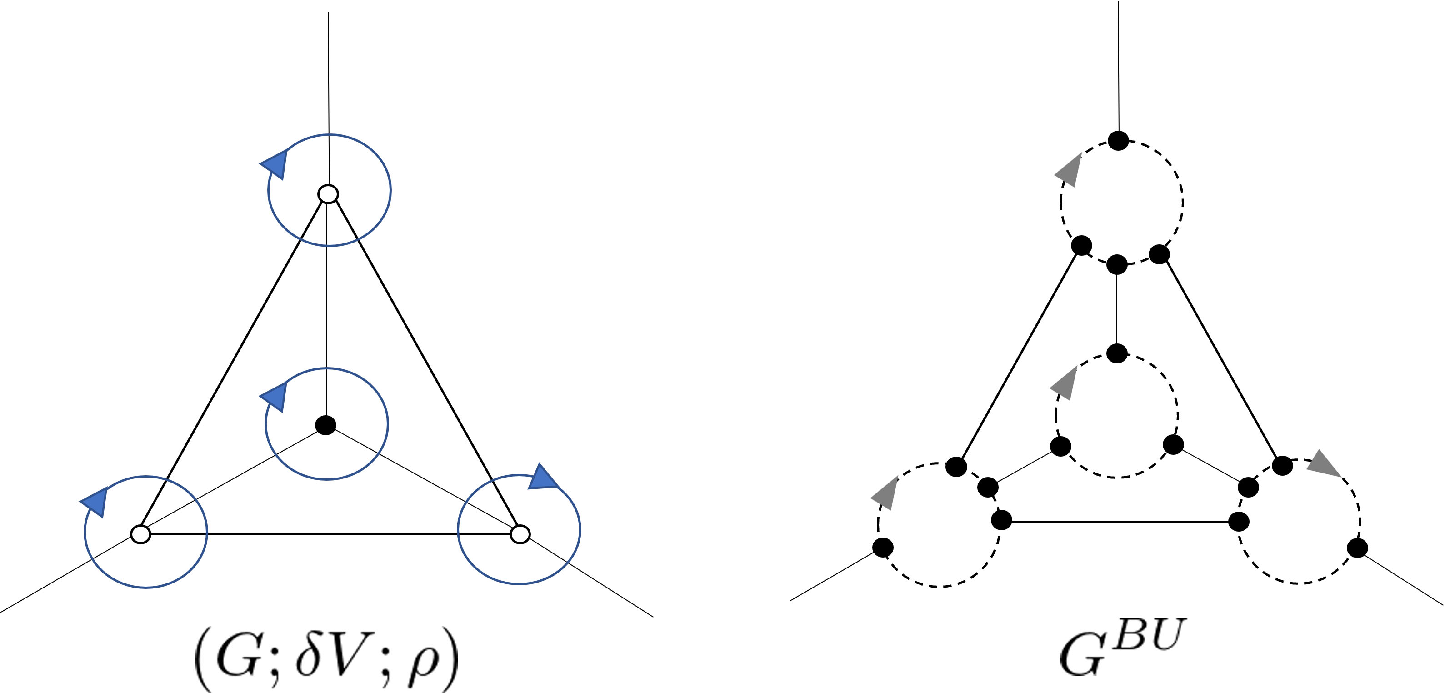}
    \caption{A rotation tailed graph $(G;\delta V;\rho)$ and its blow up graph $G^{BU}$: In the figure of $(G;\delta V;\rho)$, the white vertices are $\delta V$ and each clockwise circle describes the rotation on each vertex. In the figure of $G^{BU}$, the dotted lines are the arcs in $A_{is}$ while the real lines are  the arcs in $A_{br}$.}
    \label{fig:1}
\end{figure}
\begin{figure}[hbtp]
    \centering
    \includegraphics[keepaspectratio, width=80mm]{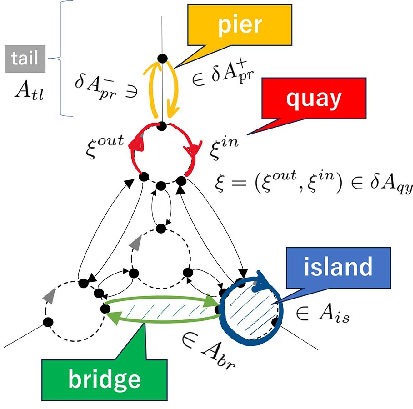}
    \caption{The names of arcs of $G^{BU}=(V^{BU},A^{BU})$. The {\it island} $A_{is}$ is the set of the oriented cycles induced by blowing up the vertices of the original graph. The {\it bridge} $A_{br}$ is the set of arcs which is isomorphic to the symmetric arc set of the original graph. The set of $A_{tl}$ is the set of arcs of the tails. The subset $\delta A_{pr}\subset A_{tl}$ called the {\it pier} is the set of arcs whose terminal or origin vertices belongs to the internal graph. The {\it quay} $\delta A_{qy}$ is the pair of island arcs $(\xi^{out},\xi^{in})$, where $t(\xi^{out})=o(\xi^{in})$ is connected to a tail. }
    \label{fig:names}
\end{figure}
The set $A^{BU}$ can be decomposed into $A_{is}$ and $A_{br}$, where $A_{is}$ is induced by the condition~$(i)$, which is the set of arcs on the oriented cycles by blowing up the original vertices, while $A_{br}$ is induced by the condition $(ii)$, which is isomorphic to the symmetric arc set of the internal of $\tilde{G}$. 
The set of the arcs obtained by blowing up the original vertex $u\in V$ is called the {\it island} of $u$, and denoted by $A_{is}(u)$. 
Each arc in $A_{br}$, which is called the {\it bridge}, connects between two islands. 
Then the arc set of $G^{BU}$ can be decomposed into 
\[ A(G^{BU})=A^{BU}\sqcup A_{tl}=A_{is} \sqcup A_{br} \sqcup A_{tl}. \]
The set of ``quays" of the blow up graph $G^{BU}$ is defined by
\[\delta A_{qy}:=\{(\xi^{out},\xi^{in})\in A_{is}\times A_{is} \;|\; t(\xi^{out})=o(\xi^{in}) \in \delta V^{BU}\}. \] 
The sets of piers $\delta A_{pr}$ and quays $\delta A_{qy}$ of $G^{BU}$ 
will be useful to describe the scattering on the surface of this quantum walk. 
See Figure~\ref{fig:names}.

Note that each vertex in $V^{BU}$ is a crossing point between each arc of $A_{is}$ and $A_{br}$; the in-degree and out-degree are $2$.
 The original motivation for such a graph deformation is an implementation of quantum walks on graphs by optical polarizing elements. At each vertex of $G^{BU}$, the incoming arcs along the island and the bridge can correspond to horizontal and vertical polarizations, respectively. The matrix  $H$, which will be  introduced in the next subsection, represents the local scattering of the polarization waves and plays the role of the half wave plate. 
 See \cite{Mizutani_ETAL} for more detail.  
\subsection{Facial quantum walk}
For each $u\in V$, we set the $2$-dimensional unitary matrix represented by  
\[H=\begin{bmatrix}a & b \\c& d\end{bmatrix},  \] 
which corresponds to the half wave plate in one of the optical elements in \cite{Mizutani_ETAL}. 
This matrix will express the local scattering at each vertex of $V^{BU}$ at each time step of the facial quantum walk.  
We assume $abcd\neq 0$ to avoid a trivial walk. 
We also assume $d\in \mathbb{R}$, and this assumption will lead a useful property of the stationary state for reflecting some interesting underlying graph structures:    
\begin{assumption}
$abcd\neq 0$ and $d\in \mathbb{R}$.
\end{assumption}
\noindent We set $\omega:=-\det H$, which will be an important factor to describe the stationary state of our quantum walk. Remark that $|\omega|=1$ by the unitarity of $H$. 
\begin{definition}[Facial quantum walk]
\noindent
\begin{enumerate}
    \item The total state space: $\mathbb{C}^{A^{BU}\cup A_{tl}}$. 
    \item Time evolution: 
    Let $\Psi_n\in \mathbb{C}^{A^{BU}\cup A_{tl}}$ be the $n$-th iteration of the quantum walk with $\Psi_{n+1}=U\psi_n$. 
    The unitary time evolution operator $U: \mathbb{C}^{A^{BU}\cup A_{tl}}\to \mathbb{C}^{A^{BU}\cup A_{tl}}$ is defined as follows.
    Let us set the arcs of a tail $e_0,e_1,\dots$ and $\bar{e}_0,\bar{e}_1,\dots$ with $t(e_0)\in \delta V^{BU}$, $t(e_{j+1})=o(e_j)$ $j=0,1,2,\dots$.  
    The time evolution is the tail is free; that is,  
\begin{align}
    \Psi_{n+1}(e_{j}) &= \Psi_n(e_{j+1}), \label{eq:tail1}\\
    \Psi_{n+1}(\bar{e}_{j+1}) &= \Psi_n(\bar{e}_j), \label{eq:tail2}
\end{align}
for any $j=0,1,2,\dots$. 
On the other hand, the time evolution in the internal graph is described as follows. 
    Let $e_{+}^{is}, e_{-}^{is}\in A_{is}$ and $e_{+}^{br},e_{-}^{br}\in A_{br}\cup A_{tl}$ be the arcs with $t(e_{+}^{is})=t(e_{+}^{br})=o(e_{-}^{is})=o(e_{-}^{br})$. 
Then  
\begin{align}\label{eq:internalTE}
    \begin{bmatrix}
    \Psi_{n+1}(e^{is}_-) \\ \Psi_{n+1}(e^{br}_-)
    \end{bmatrix}
    = H 
    \begin{bmatrix}
    \Psi_{n}(e^{is}_+) \\ \Psi_{n}(e^{br}_+)
    \end{bmatrix}.
\end{align}
\item The initial state: Let the set of the tails be $\{Tail_1,\dots,Tail_\kappa\}$. Here $\kappa=|\delta V^{BU}|$. 
The amplitudes of the inflow from the tails are set as the sequence of the complex numbers $\alpha_1,\dots,\alpha_\kappa\in \mathbb{C}$. 
The initial state $\Psi_0\in \mathbb{C}^{A^{BU}\cup A_{tl}}$ is denoted by 
\begin{align}\label{eq:initialstate}
    \Psi_0(e) = 
    \begin{cases}
    \alpha_j & \text{: $e\in A(Tail_j)$, $\dist(t(e),G)<\dist(o(e),G)$, $(j=1,\dots,\kappa)$,}\\
    0 & \text{: otherwise.}
    \end{cases}
\end{align}
\end{enumerate}
\end{definition}
By the time evolution in the internal graph (\ref{eq:internalTE}), the elements of $H$, $a,b,c,d$, can be regarded as the weights associated with the moving from $A_{is}$ to $A_{is}$, from $\bar{A}_{is}$ to $A_{is}$,  
from $A_{is}$ to $\bar{A}_{is}$,  
and from $\bar{A}_{is}$ to $\bar{A}_{is}$, respectively.  
Here $\bar{A}_{is}$ is the complement of $A_{is}$, that is, $\bar{A}_{is}=A_{br}\cup A_{tl}$. 
By the setting of the initial state~(\ref{eq:initialstate}) and the free dynamics on the tail given in (\ref{eq:tail1}) and (\ref{eq:tail2}), a quantum walker penetrates from the tails at every time step, which is regarded as the inflow to the internal graph, on the other hand, once a quantum walker goes out to the tails, it never goes back to the internal graph, which is regarded as the outflow from the internal graph. 
It is shown that $\Psi_n(e)$ converges to a stationary state for any $e\in A(G^{BU})$ in the long time limit $n\to\infty$~\cite{HS}. 
\begin{theorem}[Theorem~3.1 in \cite{HS}]
Let $\Psi_n$ be the $n$-th iternation of the quantum walk. Then we have 
\begin{align*}
\exists \lim_{n\to\infty}\Psi_n(a) &=:\Psi_\infty(a) \text{ pointwise for any $a\in A^{BU}\cup A_{tl}$;}\\
U\Psi_\infty &= \Psi_\infty. 
\end{align*}
\end{theorem}

In the rest of this section, we provide the key lemma which is the starting point of all our proofs of our theorems. 
\begin{lemma}[Key lemma]\label{lem:key}
Assume $d\in \mathbb{R}$. 
Let $\Psi\in \mathbb{C}^{A^{BU}\cup A_{tl}}$ be a generalized eigenfunction such that $U\Psi=\Psi$, and we set $\psi:=\chi\Psi$. Here $\chi$ is the indicator on $A^{BU}$. 
For connected vertices $u,u'\in V$ in the original graph $G$, let $e_{br}\in A_{br}$ be the arc connecting the islands $A_{is}(u)$ and $A_{is}(u')$ in the blow up graph $G^{BU}$. 
We set $e_{is},\epsilon_{is}\in A_{is}(u)$ and $e_{is}',\epsilon'_{is}\in A_{is}(u')$ by $t(e_{is})=o(e_{br})=o(\epsilon_{is})$ and $t(e_{is}')=t(e_{br})=o(\epsilon_{is}')$,  
respectively. 
Then we have 
\begin{align}\label{eq:key1}
    \psi(\epsilon_{is})=\omega \psi(e_{is}'),\;
    \psi(\epsilon_{is}')=\omega \psi(e_{is}) 
\end{align}
and 
\begin{align}\label{eq:key2}
    \psi(e_{br})=\frac{\omega}{b}(\psi(e_{is})+d\psi(e_{is}')),\;
    \psi(\bar{e}_{br})=\frac{\omega}{b}(\psi(e_{is}')+d\psi(e_{is})).
\end{align}
\end{lemma}
\begin{proof}
Let us put 
$\psi(e_{is})=:z_1$,
$\psi(e_{is}')=:z_1'$,
$\psi(\epsilon_{is})=:z_2$,
$\psi(\epsilon_{is}')=:z_2'$,
$\psi(e_{br})=:x$,
$\psi(\bar{e}_{br})=:x'$,
By the definition of the facial quantum walk, we have 
\begin{align}
    \begin{bmatrix} z_2 \\ x \end{bmatrix}=\begin{bmatrix}a & b \\ c & d\end{bmatrix}\begin{bmatrix} z_1 \\ x' \end{bmatrix},\;
    \begin{bmatrix} z_2' \\ x' \end{bmatrix}=\begin{bmatrix}a & b \\ c & d\end{bmatrix}\begin{bmatrix} z_1' \\ x \end{bmatrix}.
\end{align}
These equations are equivalent to  
\begin{equation}\label{eq:bridge}
    \begin{bmatrix} x \\ x '\end{bmatrix}= T \begin{bmatrix} z_1 \\ z_2 \end{bmatrix} =\sigma_X T \begin{bmatrix} z_1' \\ z_2' \end{bmatrix}, 
\end{equation}
where 
\begin{align*} 
T &= \begin{bmatrix}  0 & -b \\ 1 & -d    \end{bmatrix}^{-1} \begin{bmatrix}  a & -1 \\ c & 0 \end{bmatrix} \text{ and } \sigma_X=\begin{bmatrix}0& 1\\ 1 & 0 \end{bmatrix}.
\end{align*}
The matrix $T$ can be deformed by 
the unitarity of $H$ $(-\omega \bar{d}=a)$ and the assumption $d\in \mathbb{R}$ as follows: 
\begin{align} 
T &= \frac{1}{b}\begin{bmatrix} 1 & d \\ d & 1 \end{bmatrix} \begin{bmatrix} \omega & 0 \\ 0 & 1 \end{bmatrix}. \end{align}
From (\ref{eq:bridge}), we have 
\begin{align}
    \begin{bmatrix} z_1 \\ z_2 \end{bmatrix}
    &=\begin{bmatrix} \omega & 0 \\ 0 & 1 \end{bmatrix}^{-1}\begin{bmatrix} 1 & d \\ d & 1 \end{bmatrix}^{-1}\sigma_X\begin{bmatrix} 1 & d \\ d & 1 \end{bmatrix} \begin{bmatrix} \omega & 0 \\ 0 & 1 \end{bmatrix}\begin{bmatrix} z_1' \\ z_2' \end{bmatrix} \notag \\
    &= \begin{bmatrix} \omega^{-1} & 0 \\ 0 & 1 \end{bmatrix} \sigma_X \begin{bmatrix} \omega & 0 \\ 0 & 1 \end{bmatrix}\begin{bmatrix} z_1' \\ z_2' \end{bmatrix}, \label{eq:ensoku}  
\end{align}
which implies $z_2=\omega z_1'$ and $z_2'=\omega z_1$, where in the second equality, we used the commutativity of $\sigma_X$ and $\begin{bmatrix} 1 & d \\ d & 1 \end{bmatrix}$. Inserting this into (\ref{eq:bridge}), we obtain 
\[ x=\frac{\omega}{b}(z_1+dz_1'),\;x'=\frac{\omega}{b}(z_1'+dz_1). \]
It is completed the proof. 
\end{proof}
\begin{remark}
Without the assumption of $d\neq \mathbb{R}$, the relation corresponding to (\ref{eq:ensoku}) is 
\[ \begin{bmatrix} z_1 \\ z_2 \end{bmatrix}=\begin{bmatrix} \omega & 0 \\ 0 & 1 \end{bmatrix}^{-1}\begin{bmatrix} 
-2i \frac{\mathrm{Im}(d)}{|b|^2} & \frac{1-d^2}{|b|^2} \\ \frac{1-\bar{d}^2}{|b|^2} & 2i \frac{\mathrm{Im}(d)}{|b|^2} \end{bmatrix}\begin{bmatrix} \omega & 0 \\ 0 & 1 \end{bmatrix}\begin{bmatrix} z_1' \\ z_2' \end{bmatrix}. \]
Since the diagonal parts of the $2\times 2$ matrix in RHS remains as long as $\mathrm{Im}(d)\neq 0$, it is impossible to describe the relations of the amplitudes like the simple expression of (\ref{eq:key1}) in general. 
That's why we gave the assumption ``$d\in \mathbb{R}$" throughout this paper. 
\end{remark}
This key lemma makes us possible to obtain the values of the stationary state on $A_{is}$ sequentially along a facial walk of the graph because of (\ref{eq:key1}). 
This is the derivation of the name of our quantum walk model, {\it facial quantum walk}. 
Let us see it more precisely in the next subsection.

\section{Scattering matrix}
\subsection{Scattering matrix and the external facial walks}
Let $\delta V^{BU}=\{v_1,\dots,v_{\delta V}\}$ and $\delta A_{pr}^+=\{\epsilon_1,\dots,\epsilon_{|\delta V|}\}$ such that $t(\epsilon_j)=v_j$. 
The stationary state is denoted by $\Psi_\infty\in\mathbb{C}^{A^{BU}\cup A_{tl}}$. 
Let us represent the inflow from the outside by $\bs{\alpha}_{in}\in \mathbb{C}^{\delta V^{BU}}$ such that 
\[ \bs{\alpha}_{in}(v)=\Psi_\infty(\epsilon) \text{ with $t(\epsilon)=v$} \]
for any $v\in \delta V^{BU}$, 
while the outflow to the outside by  $\bs{\beta}_{out}\in \mathbb{C}^{\delta V}$ such that 
\[ \bs{\beta}_{out}(v)=\Psi_\infty(\bar{\epsilon}) \text{ with $t(\epsilon)=v$} \] 
for any $v\in \delta V^{BU}$. 
The scattering matrix $S:\mathbb{C}^{\delta V}\to \mathbb{C}^{\delta V}$ is defined by 
\[  \bs{\beta}_{out}=S\bs{\alpha}_{in}, \]
which is independent of the inflow $\bs{\alpha}_{in}$. 
It is shown in \cite{HS} that 
such a matrix $S$ exists and this matrix is unitary. 

Before stating our main results for the scattering matrix $S$, we prepare important graph notions which express the scattering matrix. 
Let us consider a closed walk in $G^{BU}$ represented by a sequence of {\it arcs} in $A_{br}$ and {\it walks} in $A_{is}$; 
\[ f=(e_0,\xi_0,e_1,\xi_1,\dots,e_{s-1},\xi_{s-1}), \;\;(j=0,\dots,\kappa-1) \] 
where $t(\xi_{s-1})=o(e_0)$ and 
\[e_{j+1}=
\begin{cases}
\rho(e_j) \text{ in $\tilde{G}$} & \text{: $o(\rho(e_j))\notin \delta V^{BU}$, }\\
\rho(\overline{\rho(e_{j})}) \text{ in $\tilde{G}$} & \text{: $o(\rho(e_j))\in \delta V^{BU}$,}
\end{cases} \]
\[\xi_j=\begin{cases}
\xi_j\in A_{is} \text{ with }\;o(\xi_j)=t(e_j),\;t(\xi_j)=o(e_{j+1}) & \text{: $o(\rho(e_j))\notin \delta V^{BU}$}\\
(\xi_j^{out},\xi_j^{in})\in \delta A_{qy} & \\
\qquad\qquad\text{ with } o(\xi_j^{out})=t(e_{j}),\;t(\xi_j^{out})=o(\xi_j^{in}),\;t(\xi_j^{in})=o(e_{j+1}) & 
\text{: $o(\rho(e_j))\in \delta V^{BU}$}
\end{cases}\]
The length of such a walk $f$ in $G$ is denoted by $s=|f|_G$. 
If the sequence of $f$ has no overlaps with $\delta A_{qy}$, we call it the {\it internal} facial closed walk. On the other hand, if $f$ has an overlap with $\delta A_{qy}$, we call it the {\it external} facial closed walk. The sets of all  internal closed facial walks and all external walks are denoted by $F^{in}$ and $F^{ex}$, respectively. 
See Figure~\ref{fig:2}.
\begin{figure}[hbtp]
    \centering
    \includegraphics[keepaspectratio, width=120mm]{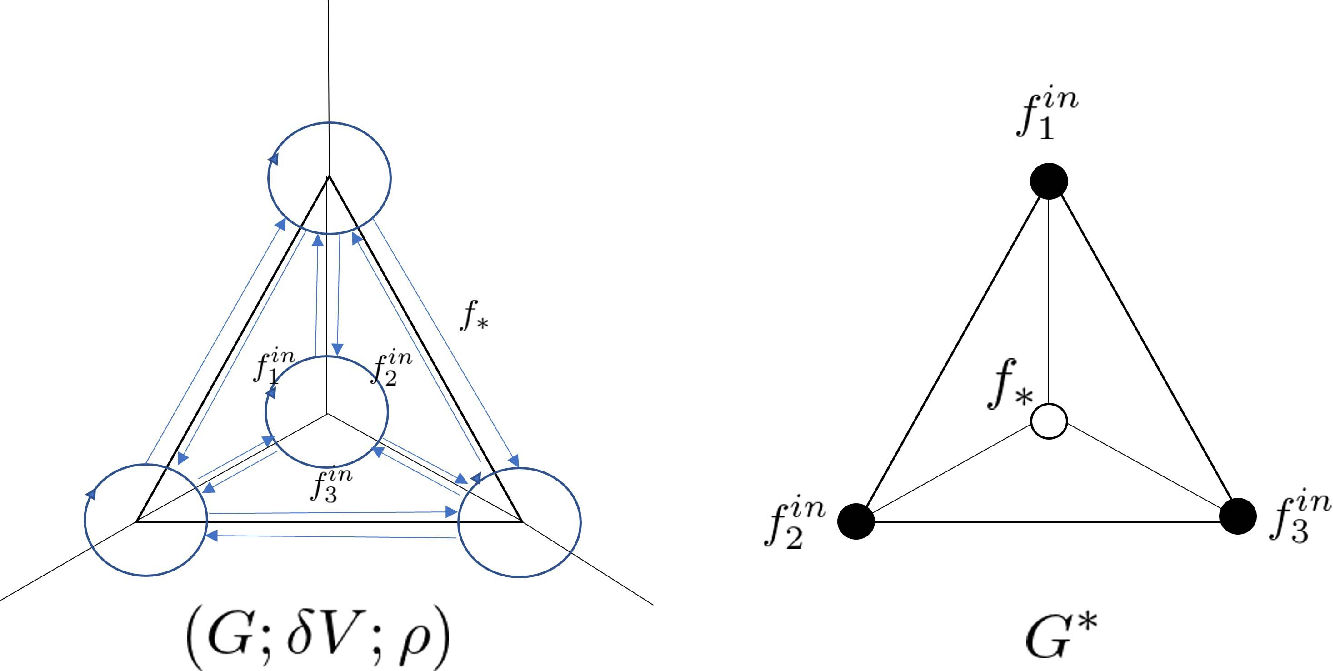}
    \caption{The internal and external facial walks and the dual graph: The closed oriented cycles correspond to facial walks in $G^{BU}$. The facial walk which passes through the tails, $f_*$, is the external facial walk, while the other facial walks, $f_1^{in},f_2^{in},f_3^{in}$, are the internal facial walk. The right figure illustrating that each facial walk in the left figure corresponds to each vertex and the external facial walk corresponds to the sink, will be used for the discussion in Section~\ref{sect:4.3.3}.  }
    \label{fig:2}
\end{figure}

Let $f\in F^{ex}$ be an external facial closed walk which visits $\kappa$ vertices of $\delta V^{BU}$: 
\begin{equation}\label{eq:externalwalk}
    f:=(\xi_{\ell_0}^{in},e_{\ell_0+1},\xi_{\ell_0+1},\dots,e_{\ell_1},\xi_{\ell_1}^{out},\dots,\xi_{\ell_{\kappa-1}}^{in},e_{\ell_{\kappa-1}+1},\xi_{\ell_{\kappa-1}+1},\dots,e_{\ell_{0}},\xi_{\ell_{0}}^{out}), 
\end{equation} 
where 
$e_{\ell_j+m}\in A_{br}$, $\xi_{\ell_j+m}\in A_{is}$ and 
$\xi_{\ell_j}=(\xi_{\ell_j}^{out},\xi_{\ell_j}^{in})\in \delta A_{qy}$ ($j=0,\dots,\kappa-1$ and $m=1,\dots, \ell_{j+1}-\ell_j$). We set $\partial f:=\{\xi_{\ell_0},\dots,\xi_{\ell_{\kappa-1}}\}$ and its cardinality by $\kappa=|\partial f|$. 
The set of $\delta f$ plays a role of a  ``quay" receiving inward from the outside and radiating outward from this external facial walk $f$. 
Recall that every walk denoted as $\xi_j\in \delta A_{qy}$ is expressed by the pair of $2$ arcs $\xi_j=(\xi_j^{out},\xi_j^{in})$ with $t(\xi_j^{out})=o(\xi_j^{in})\in \delta V^{BU}$. 
We also set the distance between boundary vertices in $G$ by 
$\delta_j(f):=\ell_{j+1}-\ell_j$ for $j=0,\dots,|\partial f|-1$ in the modulus of $|\partial f|$. 
For arbitrary external facial closed walk $f\in F^{ex}$,  we set the identity matrix, the weighted cyclic permutation matrix on $\mathbb{C}^{\{0,\dots,|\partial f|-1\}}$ by 
\begin{equation}\label{eq:op_ex_walk} (I_{f}h)(j)=h(j),\;(P_{f}(\omega)h)(j)=\omega^{\delta_{j-1}(f)}h(j-1)
\end{equation}
for any $h\in \mathbb{C}^{\{0,\dots,|\partial f|-1\}}$ and $j\in \{0,\dots,|\partial f|-1\}$ in the modulus of $|\partial f|$, respectively. Here the unit complex number $\omega$ was defined by $\omega=-\det(H)$. 
Now we are ready to give our theorem for the scattering matrix.
\begin{theorem}\label{thm:scattering}
The setting of the graph $G^{BU}$ and the facial quantum walk are denoted as the above. 
Then the scattering matrix of the facial quantum walk is decomposed into the following $|F^{ex}|$ unitary matrices as follows:  \[S=\bigoplus_{f\in F^{ex}} S_{f}, \]
where 
\[ S_{f}=bcP_{f}(\omega)\;(I_{f}-aP_{f}(\omega))^{-1}+dI_{f}. \]
Here the operators induced by each external facial closed walk $I_f$ and $P_f(\omega)$ are defined in (\ref{eq:op_ex_walk}).
\end{theorem}
\begin{proof}
For an external closed walk $f\in F^{ex}$ expressed as  (\ref{eq:externalwalk}), 
the boundaries of the islands with tails, that overlap with the walk $f$, are denoted by $u_0,\dots,u_{\kappa-1}\in \delta V^{BU}$.  
We identify each element of $\delta V^{BU}$ with that of $\delta V$. 
 Let us put $\Psi_\infty(\xi_{\ell_j}^{(in)})=:b \eta_j$ ($j=0,1,\dots,\kappa-1$).
By Lemma~\ref{lem:key}, we have 
\begin{equation*}
    \psi_{\infty}(\xi_{\ell_{j+1}}^{(out)})=b\eta_j \omega^{\delta_j}\;\;(j=0,\dots,\kappa-1). 
\end{equation*} 
Since $U\Psi_\infty=\Psi_\infty$, it holds that
\begin{equation}\label{eq:eta0} \eta_{j+1}=a\;\omega^{\delta_j}\eta_j+\bs{\alpha}_{in}(u_{j+1})\;\;(j=0,\dots,\kappa-1)  
\end{equation}
and 
\begin{align}\label{eq:beta}
\bs{\beta}_{out}(u_{j+1}) &= bc\;\omega^{\delta_{j}}\eta_j+d\;\bs{\alpha}_{in}(u_{j+1}) \;\;(j=0,\dots,\kappa-1).
\end{align}
Set $\bs{\eta}_f:=[\eta_0\dots,\eta_{\kappa-1}]^\top$, $\bs{\alpha}_{in,f}=[\bs{\alpha}_{in}(u_0),\dots,\bs{\alpha}_{in}(u_{\kappa-1})]^\top$ and $\bs{\beta}_{out,f}=[\bs{\beta}_{in}(u_0),\dots,\bs{\beta}_{in}(u_{\kappa-1})]^\top$. Then (\ref{eq:eta0}) is equivalent to 
\begin{equation}\label{eq:eta}
\bs{\eta}_f=(I_f-aP_f(\omega))^{-1}\bs{\alpha}_{in,f}, 
\end{equation}
which implies 
 \begin{align*}
     \bs{\beta}_{out,f} &= bcP_f(\omega)\bs{\eta}_f+d\bs{\alpha}_{in,f} \\
     &=(\; bcP_f(\omega)(I_f-aP_f(\omega))^{-1}+dI_f\;)\bs{\alpha}_{in,f}. 
 \end{align*} 
Then we obtain the desired conclusion. 
\end{proof}
\begin{remark}\label{rem:scattering_decom}
The matrix $I_f-aP_f(\omega)$ is invertible. 
Indeed, if $|\partial f|$ is $\kappa$, then $S_f$ can be expanded by 
\[ S_f=\frac{bc}{1-a^\kappa\Delta(f)}P_f\left( I_f+aP_f+\cdots+(aP_f)^{\kappa-1} \right)+dI_f, \]
where $\Delta(f)=\omega^{\delta_0+\cdots+\delta_{\kappa-1} }$. 
If $\omega=1$, then 
the scattering matrix $S_f$ becomes a circulant and unitary matrix. 
\end{remark}
\begin{remark}
The scattering matrix can be only described by the information on the external facial walks of the graph. Moreover to describe each scattering on the external facial walk $f\in F^{ex}$, we only need the information on geometric positions of the tails in $f$. 
\end{remark}
\subsection{Detection of the underlying $2$-cell embedding on orientable surface of $K_{3,3}$ from the scattering of facial quantum walk}\label{sect:ex1} 

Using Theorem~\ref{thm:scattering}, let us consider the three kinds of rotations to the same graph $K_{3,3}$ in Figure~\ref{fig:embedding}. 
Each rotation is denoted by $[10,4,4]$-type, $[6,6,6]$-type and $[18]$-type, respectively.  
In $[10,4,4]$-type, the induce facial walks are represented by the sequence of $10$ vertices and $2$ kinds of sequences of $4$ vertices such that 
\begin{align*}
&``\;1\to 3'\to 2\to1'\to 3\to 3'\to 1\to 2'\to 3\to 1'(\to 1)\;",\\
&``\;1\to 1'\to 2\to 2'(\to 1)\;",\\
&``\;2'\to 2\to 3'\to 3(\to 2')\;".
\end{align*} 
In $[6,6,6]$-type, the induce facial walks are represented by three kinds of sequences of $6$ vertices such that 
\begin{align*}
&``\;1'\to 1\to 3'\to 3\to 2'\to 2(\to 1')\;"\\
&``\;1\to 1'\to 3\to 3'\to 2\to 2'(\to 1)\;"\\
&``\;1\to 2'\to 3\to 1'\to 2\to 3'(\to 1)\;".
\end{align*} 
In $[18]$-type, the induce facial walk is represented by sequence of $18$ vertices such that 
\begin{multline*}``\;1'\to  1\to 2'\to 3\to 3'\to 1\to 1'\to 2\to 2'\\\to 1\to 3'\to 2\to 1'\to 3\to 2'\to 2\to 3'\to 3(\to 1')\;".
\end{multline*} 


Let us set the tails as follows. See Figure~\ref{fig:tail_embedding}. 
Assume that the tails are joined to all the vertices of $K_{3,3}$ so that every tail is placed between the vertices $1'$ and $2'$ for any vertices $1,2,3$, while the tail is placed between $1$ and $2$ for any vertices $1',2',3'$. 
Then the interaction from the tails to the scattering exhibits at the vertex between the vertices $1$ and $2$ or between the vertices $1'$ and $2'$ in each facial walk. 
For example, in the facial walk on $[10,4,4]$-type, the interaction from the tails exhibits at the bold face vertices: $(1,\bs{3'},2,1',3,3',1,2',\bs{3},1')$.  
For the other facial walks, in the facial walk on $[10,4,4]$-type, ``$(\bs{1},\bs{1'},\bs{2'},\bs{2})$", the interaction from the tails exhibits at every vertex, while in the facial walk ``$(2',2,3,3')$", there are no interaction to any tails, which is the internal facial closed walk. 
Assume $\omega=1$. By Theorem~\ref{thm:scattering}, the scattering matrix for $[10,4,4]$-type is constructed by the permutation on the vertex set $\{1,2,3,1',2',3'\}$ such 
 that $(3,3')(1',1,2',2)\in \mathfrak{S}_6$: 
 \begin{equation}\label{eq:10,4,4} 
 S:=S([10,4,4])\cong S_2 \oplus S_4.
 \end{equation}
Here 
\[S_n= \frac{bc}{1-a^n}P_n(I_n+aP_n+\cdots+a^{n-1}P_{n-1})+dI_n, \]
where $P_n$ is the permutation matrix on $\{0,\dots,n-1\}$ such that 
$(P_nf)(j)=f(j-1)$ by Remark~\ref{rem:scattering_decom}. 
The scattering for $[6,6,6]$-type is constructed by the permutation $(2,1')(1,2')(3,3')\in \mathfrak{S}_6$: 
\begin{equation}\label{eq:6,6,6} 
S([6,6,6])\cong S_2 \oplus S_2  \oplus S_2. 
\end{equation}
The scattering for the $[18]$-type is constructed by the permutation $(1,1',2,2',3',3)\in \mathfrak{S}_6$: 
\begin{equation}\label{eq:18} 
S([18])\cong S_6. 
\end{equation}
Assume that we are not informed the underlying orientable surface of the $2$-cell embedding for $K_{3,3}$. Under this situation, let us  detect the embedding way of $K_{3,3}$ by using the scattering information. 
Let us input the inflow from the vertex $1$. 
It is easily check that every entry of $S_n$ is non-zero. Then by (\ref{eq:10,4,4}), (\ref{eq:6,6,6}) and (\ref{eq:18}), the support vertices of the outflows are described by
\begin{align*}
\supp(\beta_{out}([10,4,4])) &= \{1,2,1',2'\}, \\
\supp(\beta_{out}([6,6,6])) &= \{1,2'\}, \\
\supp(\beta_{out}([18])) &= \{1,2,3,1',2',3'\}, 
\end{align*}
respectively. 
The number of tails where the outflow is detected is denoted by $N$. Then we can see that the number of detected tails of quantum walker has the bijection to the underlying embedding as follows:  
\[ \text{The underlying embedding of $K_{3,3}$}=\begin{cases} [10,4,4] & \text{, if $N=4$,}\\ [6,6,6] & \text{, if $N=2$,}\\ [18] & \text{, if $N=6$.}\end{cases} \]
Note that this setting of the configuration of tails and inflow were for the purpose of the detection of the underlying embedding on the closed surface of $K_{3,3}$, but adjusting them, we can obtain another scattering.  
\begin{figure}[hbtp]
    \centering
    \includegraphics[keepaspectratio, width=160mm]{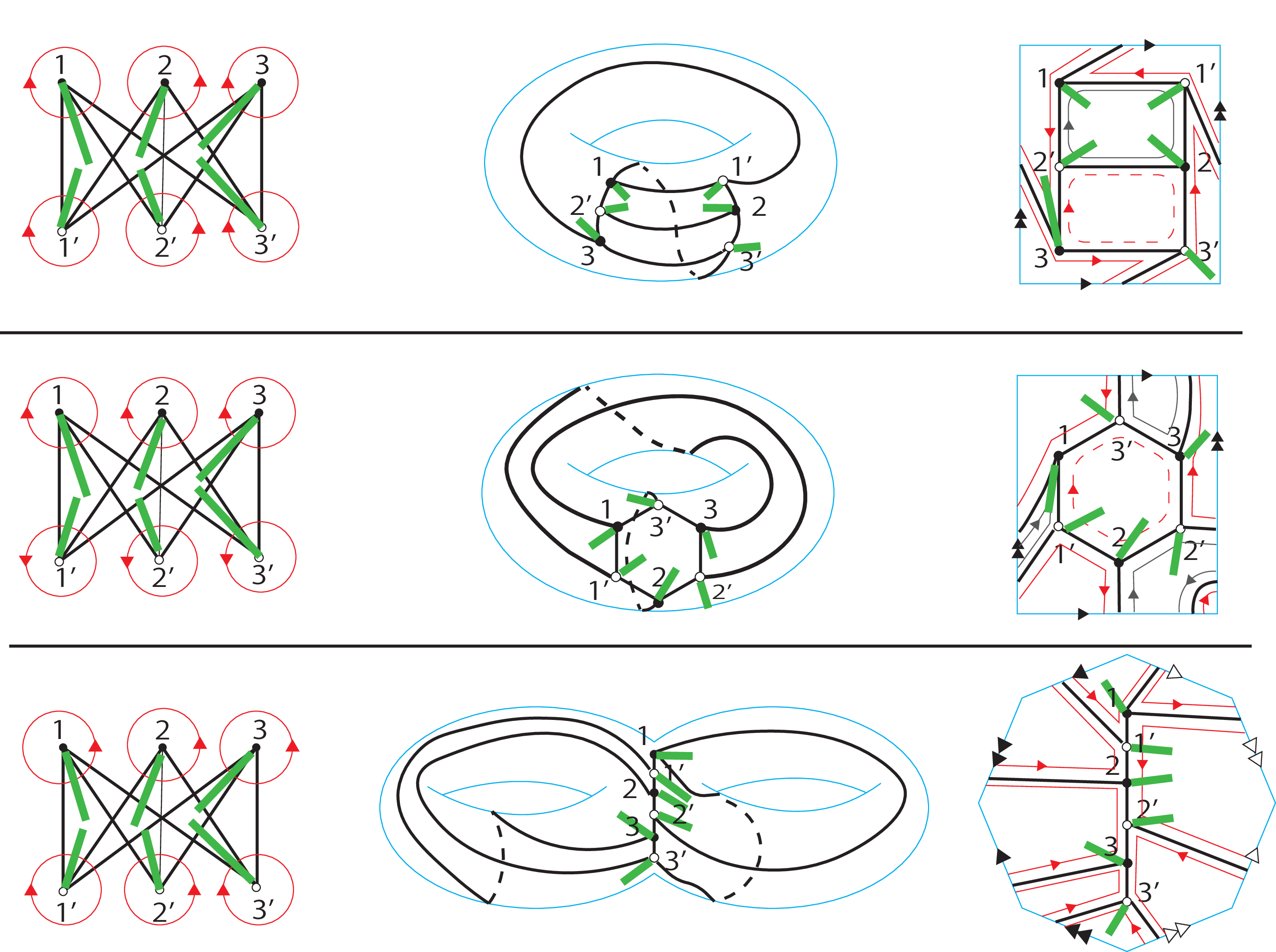}
    \caption{The location of tails in each embedding of $K_{3,3}$ for Section~\ref{sect:ex1}: The bold (green) lines depict the tails in the setting of Section~\ref{sect:ex1}. 
    The inflow comes from only the tail joining to the vertex $1$. Let us call this tail an incoming tail. 
    For the embedding [4,4,10], the number of tails which is included in the same face as the incoming tail is $N=4$; 
    for the embedding [6,6,6], the number of tails which is included in the same face as the incoming tail is $N=2$; 
    for the embedding [18], the number of tails which is included in the same face as the incoming tail is $N=6$. }
    \label{fig:tail_embedding}
\end{figure}

\section{ Stationary state}
In this section, we consider the stationary state and how a quantum walker penetrates the internal graph using some graph geometric information; in particular the facial closed walks. Since our interest focus on the stationary state in the internal graph, we introduce the operator which restricts the functions to the internal graph as follows. Let $\chi: \mathbb{C}^{A^{BU}\cup A_{tl}}\to \mathbb{C}^{A^{BU}}$ be the restriction of any function $\Psi\in \mathbb{C}^{A^{BU}\cup A_{tl}}$ to $\mathbb{C}^{A^{BU}}$ such that 
\[ (\chi\Psi)(e) = \Psi(e) \]
for any $e\in A^{BU}$. 
The adjoint $\chi^*$ can be described by 
\[ (\chi^*\psi)(e) = \begin{cases} \psi(e) & \text{: $e\in A^{BU}$} \\
0 & \text{: otherwise.}\end{cases} \]
for any $\psi\in \mathbb{C}^{A^{BU}}$ and $e\in A^{BU}\cup A_{tl}$. The restriction of $\Psi_\infty\in \mathbb{C}^{A^{BU}\cup A_{tl}}$ to the internal $\mathbb{C}^{A^{BU}}$ is denoted by $\psi_\infty:=\chi \Psi_\infty$. 
\subsection{Eigenvectors induced by facial closed walks}
Let us start by considering the following lemma which provides a useful approach to compute the stationary state. 
\begin{lemma}[\cite{HS}]\label{lem:HS}
Let $\delta A^+_{pr}$ be $\{\epsilon_1,\dots,\epsilon_{|\delta V|}\}$. 
The stationary state in the uniformly bounded functional space is the unique solution of the generalized eigenequation  $U\Psi_\infty=\Psi_\infty$ satisfying with the following conditions.  
\[ \Psi_\infty(\epsilon_j)=\bs{\alpha}_{in}(j) 
\text{ for any $j=1,\dots,|\delta V|$} \]
and  
\[ \chi\Psi_\infty\in \ker(1-\chi U\chi^* )^\perp. \]
\end{lemma}
\noindent The above lemma yields that the stationary state is obtained by removing the $(+1)$-eigenvectors of $U$, whose supports are included in the internal graph $A^{BU}$, from a generalized eigenfunction satisfying the boundary conditions.   
In the following, we identify such $(+1)$-eigenvectors and a generalized eigenfunction by using an expression of the facial closed walks. 

Let us explain the derivation of the internal and external facial walks denoted by the following Definition~\ref{def:facialfunction}. 
See also Figure~\ref{fig:soccer1} which illustrates the internal and external facial walks for the truncated icosahedron in Section~4.3.1. 
The internal facial walk induced by \[f=(e_0,\xi_0,\dots,e_{\kappa-1},\xi_{\kappa-1})\in F^{in}\] is constructed by tracing the walk $f$ launched from $\xi_0$ with  $\gamma_f^{in}(\xi_0)=b$ so that the conditions (\ref{eq:key1}) and (\ref{eq:key2}) in Lemma~\ref{lem:key} are satisfied for any $a\in A(f)$ and $\bar{a}\in A_{br}(f)$ except the final place between $(\xi_{\kappa-1},e_0,\xi_0)$ in general. To satisfy the condition at all the arcs of $f$, it must be $\gamma(\xi_0)=\omega\gamma_f^{in}(\xi_{\kappa-1})$, which implies $\omega^{|f|_G}=1$. Then for any $f\in F^{in}$, we have $\gamma_f^{in}\in \ker(1-\chi U\chi^*)$ if and only if $\omega^{|f|_G}=1$. 
On the other hand, the derivation of the external facial walk is constructed so that $\gamma_f^{ex}$ satisfies the conditions (\ref{eq:key1}) and (\ref{eq:key2}) at every arcs of every partial walk of $f\in F^{ex}$;  \[(\xi_{\ell_m}^{in},e_{\ell_m+1},\xi_{\ell_m+1},e_{\ell_m+2},\xi_{\ell_m+2}\dots,e_{\ell_{m+1}},\xi_{\ell_{m+1}}^{out})\subset f\] by putting $\gamma_f^{ex}(\xi_{\ell_m}^{in})=b\eta_m$ for $m=0,1,\dots,{\kappa-1}$. 
Then for example, $\gamma_f^{ex}(e_{\xi_{\ell+1}^{out}})$ must be  $\omega^{\delta_m(f)}\times b\eta_m$ by condition (\ref{eq:key1}). 
Moreover at every quay with $A_{is}(f)\cap A_{qy}$, the following local stationary condition between the outflow and inflow must be satisfied such that 
\begin{align*} \bs{\beta}_{out}(i_{m+1})&=c\gamma_f^{ex}(\xi_{\ell_{m+1}}^{out})+d\bs{\alpha}_{in}(i_{m+1})\\
&=bc\omega^{\delta_m(f)}\eta_m +d\bs{\alpha}_{in}(i_{m+1}).
\end{align*} 
See Figure~\ref{fig:internalfacialfunction} for the internal facial function, and Figure~\ref{fig:externalfacialfunction} for the external facial function. 
\begin{definition}[Functions induced by facial closed walks]\label{def:facialfunction}
\noindent
\begin{enumerate}
\item Internal facial function: \\
For an internal facial closed walk \[f=(e_0,\xi_0,\dots,e_{s-1},\xi_{s-1})\in F^{in},\] 
with $e_j\in A_{br}$, $\xi_j\in A_{is}$ $(j=0,1,\dots,s-1)$, 
the function $\gamma_f^{in}\in \mathbb{C}^{A^{BU}}$ with $\supp(\gamma_f^{in})=A(f)\cup  \left( A(\bar{f})\cap A_{br} \right)$ is defined by \[\gamma_f^{in}=\gamma_f^{in,+}+\gamma_f^{in,-},\] where
\begin{align*} &\gamma_f^{in,+}(\xi_j)=b\omega^{j},\;\gamma_f^{in,+}(e_j)=\omega^j, \;(j=0,\dots,s-1)\;, \gamma_f^{in,+}(e)=0 \text{ for any $e\notin A(f)$}\\
&\gamma_f^{in,-}(\bar{e}_j)=d\omega^j, \;(j=0,\dots,s-1)\;, \gamma_f^{in,-}(e)=0 \text{ for any $\bar{e}\notin A(f)$}
\end{align*}
\item External facial function: \\
Let an external facial closed walk be denoted by 
\[ f=(\xi_{\ell_0}^{in},e_{\ell_0+1},\xi_{\ell_0+1},\dots, e_{\ell_0+\delta_0(f)},\xi_{\ell_1}^{out},\dots, \xi_{\ell_{\kappa-1}}^{in},e_{\ell_{\kappa-1}+1},\xi_{\ell_{\kappa-1}+1},\dots, e_{\ell_{\kappa-1}+\delta_{\kappa-1}(f)},\xi_{\ell_0}^{out}) \]
with $e_{\ell_j+m}\in A_{br}$, $\xi_{\ell_j+m}\in A_{is}$ and  $\xi_{\ell_j}=(\xi_{\ell_j}^{out},\xi_{\ell_j}^{in})\in \delta A_{qy}$ ($m=1,\dots,\delta_j(f)$; $j=0,\dots,\kappa-1$), 
where joined tails are labeled by  $Tail_{\ell_0},\dots,Tail_{\ell_{\kappa-1}}$ and placed in $\xi_{\ell_0},\dots, \xi_{\ell_{\kappa-1}}$.  
The induced external facial function $\gamma_f^{ex}\in \mathbb{C}^{A^{BU}}$ with $\supp(\gamma_f^{ex})=A(f)\cup (A(\bar{f})\cap A_{br})$ is defined by 
\[\gamma_f^{ex}:=\gamma_f^{ex,+}+\gamma_f^{ex,-}.\]
Here for $m=0,\dots,\kappa-1$, 
\begin{align*}
&\gamma_f^{ex,+}(\xi_{\ell_m}^{in}) = b\eta_m,\;
\gamma_f^{ex,+}(\xi_{\ell_m+j}) =b\omega^j\eta_m \;(j=1,\dots,\delta_m-1),\;
\gamma_f^{ex,+}(\xi_{\ell_{m+1}}^{out})=  b\omega^{\delta_m}\eta_m. \\
&\gamma_f^{ex,+}(e_{\ell_m+j})= \omega^j\eta_m,
\;(j=1,\dots, \delta_m),\;
\gamma_f^{ex,+}(e) =0 \text{ for any $e\notin A(f)$}; \\
&\gamma_f^{ex,-}(\bar{e}_{\ell_m+j}) = d \omega^{j}\eta_m,\;(j=1,\dots, \delta_m),\; \gamma_f^{ex,-}(e) =0 \text{ for any $\bar{e}\notin A(f)$,} 
\end{align*}
where $\eta_m$ $(m=0,\dots,\kappa-1)$ is denoted by
\[ \eta_m=\frac{\omega^{-\delta_{m}(f)}}{bc}\left(\bs{\beta}_{out}(\ell_{m+1})-d\bs{\alpha}_{in}(\ell_{m+1})\right). \]
The extension of $\gamma_f^{ex}\in \mathbb{C}^{A^{BU}}$ to $\tilde{\gamma}_f^{ex}\in \mathbb{C}^{A^{BU}\cup A_{tl}}$ is defined by 
\[ 
\tilde{\gamma}^{ex}_f(e)=
\begin{cases}
\gamma^{ex}_f(e) & \text{: $e\in A(f)$} \\
\bs{\alpha}_{in}(i_m) & \text{: $e\in A(Tail_{i_m})$, $\dist(o(e),G)<\dist(t(e),G)$ $(m=0,\dots,\kappa-1)$} \\
\bs{\beta}_{out}(i_m) & \text{: $e\in A(Tail_{i_m})$, $\dist(o(e),G)>\dist(t(e),G)$ $(m=0,\dots,\kappa-1)$} \\
0 &\text{: otherwise.}
\end{cases} 
\]
\end{enumerate}
\end{definition}
\begin{remark}
The internal facial function $\gamma_f^{in}$ is determined uniquely up to a multiple constant, while the external facial function $\gamma_f^{out}$ is determined uniquely. 
\end{remark}

Note that all the arcs of $A^{BU}$ are covered by all the facial walks. 
Here, the pair of an arc and its inverse arc in $A_{br}$ may appears in the same facial walk, that is, for such a facial walk $f=(e_0,\xi_0\dots,e_{\kappa-1},\xi_{\kappa-1})$, there are $i$, $j$ such that $e_i=\bar{e}_j\in A_{br}$ (see Figure\ref{fig:embedding} for [10,4,4]-type and [18]-type, for examples). In this case, we set $\gamma_f^{in}(e_i)=\gamma_f^{in,+}(e_i)+\gamma_f^{in,-}(\bar{e}_j)=\omega^i+d\omega^j$. 
\begin{figure}[hbtp]
    \centering
    \includegraphics[keepaspectratio, width=80mm]{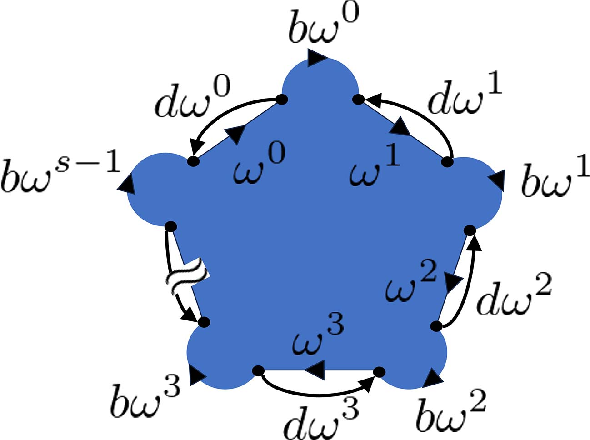}
    \caption{Internal facial function: An internal facial walk $f\in F^{in}$ whose length is $s$ is picked up. The value assigned at each arc of its internal facial function $\gamma^{in}_f$ is given as depicted by this figure. The values on the island arcs are $b\omega^0,b\omega^1,\dots,b\omega^{s-1}$, the values on the bride  arcs are $\omega^0,\omega^1,\dots,\omega^{s-1}$ and $d\omega^0,d\omega^1,\dots,d\omega^{s-1}$. }
\label{fig:internalfacialfunction}
\end{figure}
\begin{figure}[hbtp]
    \centering
\includegraphics[keepaspectratio, width=180mm]{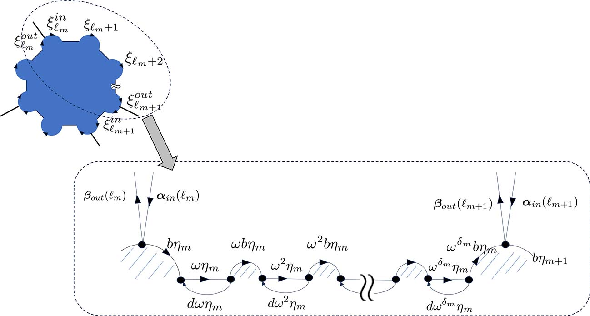}
    \caption{External facial function: The left figure shows an external facial walk $f\in F^{ex}$. The structure between the tails $Tail_{\ell_{m}}$ and $Tail_{\ell_{m+1}}$ is enlarged in the right figure. The values at each arc in this region of the induced external facial function $f$ are depicted. 
    }
\label{fig:externalfacialfunction}
\end{figure}

Now we are ready to state the following lemma.  
\begin{lemma}\label{lem:inexfunctions}
Let $\gamma^{in}_f$ and $\gamma^{ex}_f$ be defined as the above. 
Assume $d\in \mathbb{R}$. 
For any $f\in F^{ex}$,  
\[ U\tilde{\gamma}^{ex}_f=\tilde{\gamma}^{ex}_f \]
and 
\[ \ker(1-\chi U\chi^*)=\spann\{\gamma_f^{in} \;|\; f\in F^{in} \text{ such that $\omega^{|f|_G}=1$ }\}.  \]
\end{lemma}
\begin{proof}
For any $f\in F^{ex}$, 
it can be directly confirmed that  $U\tilde{\gamma}_f^{ex}=\tilde{\gamma}_f^{ex}$ by the definition of $U$. 
Let us show the second statement. 
Assume there exists $\psi\in \ker(1-\chi U\chi^*)$ such that $\supp(\psi)\cap \delta A_{qy}\neq \emptyset$. 
Since $\psi\in \ker(1-\chi U\chi^*)$, we have $||\psi||^2=||\chi U\chi^*\psi||^2$. 
Let the set of arcs $A^{BU}$ be divided into
\begin{align*} 
A^{BU} 
&= 
\bigsqcup_{u\notin \delta V^{BU}} \{ \xi^{+}_u,e^{+}_u \} \sqcup \bigsqcup_{u\in \delta V^{BU}} \{\xi^{out}_u\}  \\
&=  
\bigsqcup_{u\notin \delta V^{BU}} \{ \xi^{-}_u,e^{-}_u \} \sqcup \bigsqcup_{u\in \delta V^{BU}} \{\xi^{in}_u\},
\end{align*}
where $t(\xi_u^{+})=t(e_u^{+})=o(\xi_u^{-})=o(e_u^{-})=u$ for any $u\notin \delta V^{BU}$, and $t(\xi_u^{out})=o(\xi_u^{in})=u$ for any $u\in \delta V^{BU}$. 
Then it holds that 
\begin{align*}
\begin{bmatrix}
(\chi U \chi^* \psi)(\xi_u^{-})\\(\chi U \chi^*\psi)(e_u^{-})
\end{bmatrix}
&= H \begin{bmatrix}
\psi(\xi_u^{+})\\ \psi(e_u^{+})
\end{bmatrix}
\text{ for any $u\notin \delta V^{BU}$}
\end{align*}
and 
\begin{align*}
(\chi U \chi^*\psi)(\xi_u^{in}) &= a\psi(\xi_u^{out})
\text{ for any $u\in \delta V^{BU}$}.
\end{align*}
Then, since $H$ is unitary and $|a|<1$, the norm $\chi U\chi^*\psi$ can be evaluated by
\begin{align*}
||\chi U\chi^*\psi||^2 &= \sum_{u\in \delta V^{BU}} |a|^2|\psi(\xi^{out}_u)|^2 + \sum_{u\notin \delta V^{BU}} 
\bigg|\bigg|\; \begin{bmatrix}
\psi(\xi_u^{+})\\ \psi(e_u^{+})
\end{bmatrix} \;\bigg|\bigg|^2 \\
& < ||\psi||^2,
\end{align*}
which is the contradiction.  
Then for any $\psi\in \ker(1-\chi U\chi^*)$,  $\supp(\psi)\cap A_{qy}=\emptyset$. 
Now let us show that any eigenvector $\psi\in \ker(1-\chi U\chi^*)$ is expressed by a linear combination of $\gamma_f^{in}$'s with $\omega^{|f|_G}=1$. 
Note that there are no arcs $e_{br}$'s in $\supp(\psi)\cap A_{br}$ such that $\psi(e_{is})=\psi(e_{is}')=\psi(\epsilon_{is})=\psi(\epsilon_{is}')=0$ by Lemma~\ref{lem:key}. 
This means that every arc $e_{br}\in \supp(\psi)\cap A_{br}$ is sandwiched by arcs in  $\supp(\psi)\cap A_{is}$. 
Then let us set an island of the arc which has an overlap with the internal facial walk $f=(e_0,\xi_0,\dots,e_{\kappa-1},\xi_{\kappa-1})$ by $\xi_j\in \supp(\psi)\cap A_{is}(f)$. 
The definition of the internal facial function implies that for any $f,g\in F^{in}$ with $f\neq g$, we have $(\supp (\gamma_f^{in}) \cap A_{is})\cap (\supp (\gamma_g^{in}) \cap A_{is})=\emptyset$.
By (\ref{eq:key1}) in Lemma~\ref{lem:key}, 
this implies that if $\supp(\psi)\cap A_{is}(f)\neq \emptyset$ for some $f\in F^{in}$, then $\psi(\xi_{j+1})=\omega\psi(\xi_{j})$, $\psi(\xi_{j+2})=\omega^2\psi(\xi_{j})$,\dots, $\psi(\xi_{j-1})=\omega^{|f|_{g}-1}\psi(\xi_{j})$.  Thus $\supp(\psi)\cap A_{is}(f)$ must be $A_{is}(f)$,  and also $\omega^{|f|_G}$ must be $1$. Then for any $e_{is}\in \supp(\psi)\cap A_{is}(f)$, 
we can express $\psi(e_{is})=\alpha_f \gamma_f^{in}(e_{is})$ with some constant $\alpha_f\in \mathbb{C}$ which is independent of $e_{is}$. 
This means that $\psi|_{A_{is}(f)}=\alpha_f \gamma_f^{is}|_{A_{is}(f)}$ for any $f\in F^{is}$. Here if  $\supp(\psi)\cap A_{is}(f)= \emptyset$ or $\omega^{|f|_G}=1$, then we set $\alpha_f=0$.  From (\ref{eq:key2}) in Lemma~\ref{lem:key}, the value of $\psi(e_{br})$ with $e_{br}\in A_{br}(f)\cap A_{br}(\bar{g})$, where $A_{is}(f)\cap \supp(\psi),\;A_{is}(g)\cap \supp(\psi)\neq \emptyset$,  
must be 
\begin{equation}\label{eq:1}
    \psi(e_{br})=\frac{\omega}{b} \psi(e_{is})+d\frac{\omega}{b}\psi(e_{is}')
    =
    \begin{cases}
    \alpha_f\gamma_f^{in}(e_{br})+\alpha_g\gamma_g^{in}(\bar{e}_{br}) & \text{: $f\neq g$} \\
    \alpha_f\gamma_f^{in}(e_{br}) & \text{: $f=g$}
    \end{cases}
    , 
\end{equation} 
and 
\begin{equation}\label{eq:2}
    \psi(\bar{e}_{br})=d\frac{\omega}{b}\psi(e_{is})+\frac{\omega}{b}\psi(e'_{is})= 
    \begin{cases}
    \alpha_f\gamma_f^{in}(\bar{e}_{br})+\alpha_g\gamma_g^{in}(e_{br}) & \text{: $f\neq g$}\\
    \alpha_f\gamma_f^{in}(\bar{e}_{br}) & \text{: $f=g$}
    \end{cases}
\end{equation}
Here if $\omega^{|g|_G}\neq 0$, then $\psi(e_{is}')=0$ which means $\alpha_g=0$. 
Thus from (\ref{eq:1}) and (\ref{eq:2}), we obtain 
the eigenspace $\ker(1-\chi U\chi^*)$ is described by the linear combination of $\gamma_f^{in}$'s, which implies the desired conclusion.  
\end{proof}
Then we immediately obtain the following Proposition.
\begin{proposition}\label{prop:support}
Let $\gamma_f^{in}$ and $\gamma_g^{ex}$ be defined as the above. Then the stationary state $\Psi_\infty$ is described by 
\begin{equation}\label{eq:st}
    \Psi_\infty=\sum_{g\in F^{ex}} \tilde{\gamma}_g^{ex}-\sum_{f\in F^{in},\; \omega^{|f|_G}=1} c_f \chi^*\gamma_f^{in}
\end{equation} 
with some coefficients $c_f$'s. 
In particular, 
\begin{equation}\label{eq:luminous}
\supp(\Psi_\infty)\cap (A_{is}\setminus \delta A_{qy}) \subset \{A_{is}(f)\;|\; f\in F^{in} \text{ such that }\omega^{|f|_{G}}= 1 \}. 
\end{equation}
\end{proposition}
\begin{proof}
Note that $\tilde{\gamma}^{ex}_f$ satisfies the boundary condition in Lemma~\ref{lem:HS}, but in general, the second condition in Lemma~\ref{lem:HS} does not hold. 
Also note that $\sum_{g\in F^{ex}}\tilde{\gamma}^{ex}_g$ is only the function which has the overlap to the quays and tails. 
Let $\Gamma^{in}$ be denoted by 
\[ \Gamma^{in}=\mathrm{span}\{\chi^*\gamma_f^{in} \;|\; f\in F^{in},\;\omega^{|f|}=1\}. \]
The projection onto $\Gamma_{in}$ is denoted by $\Pi_{\Gamma^{in}}$. Then from Lemmas~\ref{lem:HS} and \ref{lem:inexfunctions}, the stationary state can be obtained 
\begin{align*}
\psi_\infty=\sum_{g\in F^{ex}}(1-\Pi_{\Gamma^{in}})\gamma_{g}^{ex}.
\end{align*}
By the Gram-Schmidt procedure to the $\{\gamma_f^{in} \;|\; f\in F^{in}\}$, RHS can be expressed by $\sum_{g\in F^{ex}}\gamma_{g}^{ex}-\sum_{c_f}c_f\gamma_f^{in}$, where the coefficients $c_f$' are determined by the Gram-Schmidt procedure. 
Then we obtain the desired conclusion. 
\end{proof}
The luminous face is the face the one whose boundary walk is $f$ such that $c_f\neq 0$ in  (\ref{eq:st}) in Proposition~\ref{prop:support}.  All the arcs of the islands of the luminous face are included on the support of the stationary state. Note that if a face is not luminous, then no arcs of the islands are on the support.
\subsection{Soccer ball induced by quantum walk}
Let us enjoy the statement of Proposition~\ref{prop:support} by the soccer ball. 
Figure~\ref{fig:soccer1} is helpful in understanding the following calculations. 
Figure~\ref{fig:soccer1} depicts the blow up of the rotation tailed truncated icosahedron by the planar embedding.  The rotation is assigned clockwise at each vertex in the planar drawing of Figure~\ref{fig:soccer1}. Therefore the facial walks correspond to the boundaries of the faces of this planner graph and in particular the outer area corresponds to the external closed walk, $f_*$, 
which has the tails. 
The external facial walk $f_*$  
is depicted by the red closed walk in Fig~\ref{fig:soccer1} described by the sequence of arcs $(\xi_0^{in},e_1,\xi_1^{out},\xi_1^{in},e_2,\xi_2^{out},\dots, \xi_5^{in},e_0,\xi_0^{out})$. The external facial function of $f_*$ is given by  
\begin{align*}
\gamma_{f_*}^{ex}(\xi_0^{in}) &= b\eta_0, 
\gamma_{f_*}^{ex}(e_1) = \omega \eta_0, 
\gamma_{f_*}^{ex}(\bar{e}_1) = d\omega \eta_0, 
\gamma_{f_*}^{ex}(\xi_1^{out}) = b\omega \eta_0,\\
\gamma_{f_*}^{ex}(\xi_2^{in}) &= b\eta_1, 
\gamma_{f_*}^{ex}(e_2) = \omega \eta_1, 
\gamma_{f_*}^{ex}(\bar{e}_2) = d\omega \eta_1, 
\gamma_{f_*}^{ex}(\xi_2^{out}) = b\omega \eta_1, \\
\cdots \\
\gamma_{f_*}^{ex}(\xi_5^{in}) &= b\eta_5, 
\gamma_{f_*}^{ex}(e_0) = \omega \eta_5, 
\gamma_{f_*}^{ex}(\bar{e}_0) = d\omega \eta_5, 
\gamma_{f_*}^{ex}(\xi_0^{out}) = b\omega \eta_5.
\end{align*}
If the inflow $\bs{\alpha}_{in}=[1,1,1,1,1,1]^\top$, let us see $\eta_j$ $(j=0,1,\dots,5)$ can be computed by 
\[ \eta_j= \frac{1}{1-a\omega}. \]
To compute $\eta_j$ $(j=0,\dots,5)$, we need to obtain the outflow $\bs{\beta}_{out}$ for the inflow $\bs{\alpha}_{in}$. To this end, let us refer Theorem~\ref{thm:scattering}. 
For the external facial walk $f_*$, 
\[ P_{f_*}(\omega)\cong \omega
\begin{bmatrix}
0&0&0&0&0&1\\
1&0&0&0&0&0\\
0&1&0&0&0&0\\
0&0&1&0&0&0\\
0&0&0&1&0&0\\
0&0&0&0&1&0\\
\end{bmatrix}. \]
Then the scattering matrix can be computed by 
\begin{align*} 
S &= bc P_{f_*}(\omega)(I_{f_*}-aP_{f_*}(\omega))^{-1}+dI_{f_*} \\
&= \frac{\omega bc}{1-(a\omega)^6}
\begin{bmatrix}
(a\omega)^5&(a\omega)^4&(a\omega)^3&(a\omega)^2&(a\omega)^1&1 \\
1&(a\omega)^5&(a\omega)^4&(a\omega)^3&(a\omega)^2&(a\omega)^1 \\
(a\omega)^1&1&(a\omega)^5&(a\omega)^4&(a\omega)^3&(a\omega)^2 \\
(a\omega)^2&(a\omega)^1&1&(a\omega)^5&(a\omega)^4&(a\omega)^3 \\
(a\omega)^3&(a\omega)^2&(a\omega)^1&1&(a\omega)^5&(a\omega)^4 \\
(a\omega)^4&(a\omega)^3&(a\omega)^2&(a\omega)^1&1&(a\omega)^5 
\end{bmatrix}
+dI_{6}.
\end{align*}
Let us set the initial state by $\bs{\alpha}_{in}=[1,1,1,1,1,1]^\top$. 
The outflow $\bs{\beta}_{out}$ is described by 
\begin{align*}
\bs{\beta}_{out}[i] &= (S \bs{\alpha}_{in})[i] \\
&= \frac{\omega bc }{1-a\omega}+d
\end{align*}
for any $i=0,\dots,5$. 
Then we have 
\[ \eta_j= \frac{\omega^{-1}}{bc}(\bs{\beta}_{out}[j+1]-d\bs{\alpha}_{in}[j+1])=\frac{1}{1-a\omega}. \]

Let us consider the case for the setting of the quantum coin $H$ so that $\omega^6=1$ but $\omega^5\neq 1$, for example, $\omega=\pi$. According to Proposition~\ref{prop:support}, the polygon of the internal facial function in Definition~\ref{def:facialfunction} whose boundary walk $f$ such that $c_f\neq 0$ must be the hexagon; the internal facial functions of the pentagons have no contribution to the stationary state. 
Then the hexagons are nothing but the luminous faces. 
In Figure~\ref{fig:soccer2}, a soccer ball pattern is depicted by coloring the luminous faces. On the other hand, if we set the quantum coin $H$ by 
$\omega^6\neq 1$ and $\omega^5=1$, for example $\omega=e^{2\pi/5}$, the luminous faces are the external facial walk and the three pentagons which are adjacent to the external facial walk, while if   
$\omega^6\neq 1$ and $\omega^5\neq 1$ for example $\omega=e^{i\pi/4}$, then the luminous face is only the external facial walk. Moreover if we set the quantum coin $H$ by $\omega^6=\omega^5=1$, for example $\omega=1$,  then all the faces are luminous faces. 
\begin{figure}[hbtp]
    \centering
    \includegraphics[keepaspectratio, width=130mm]{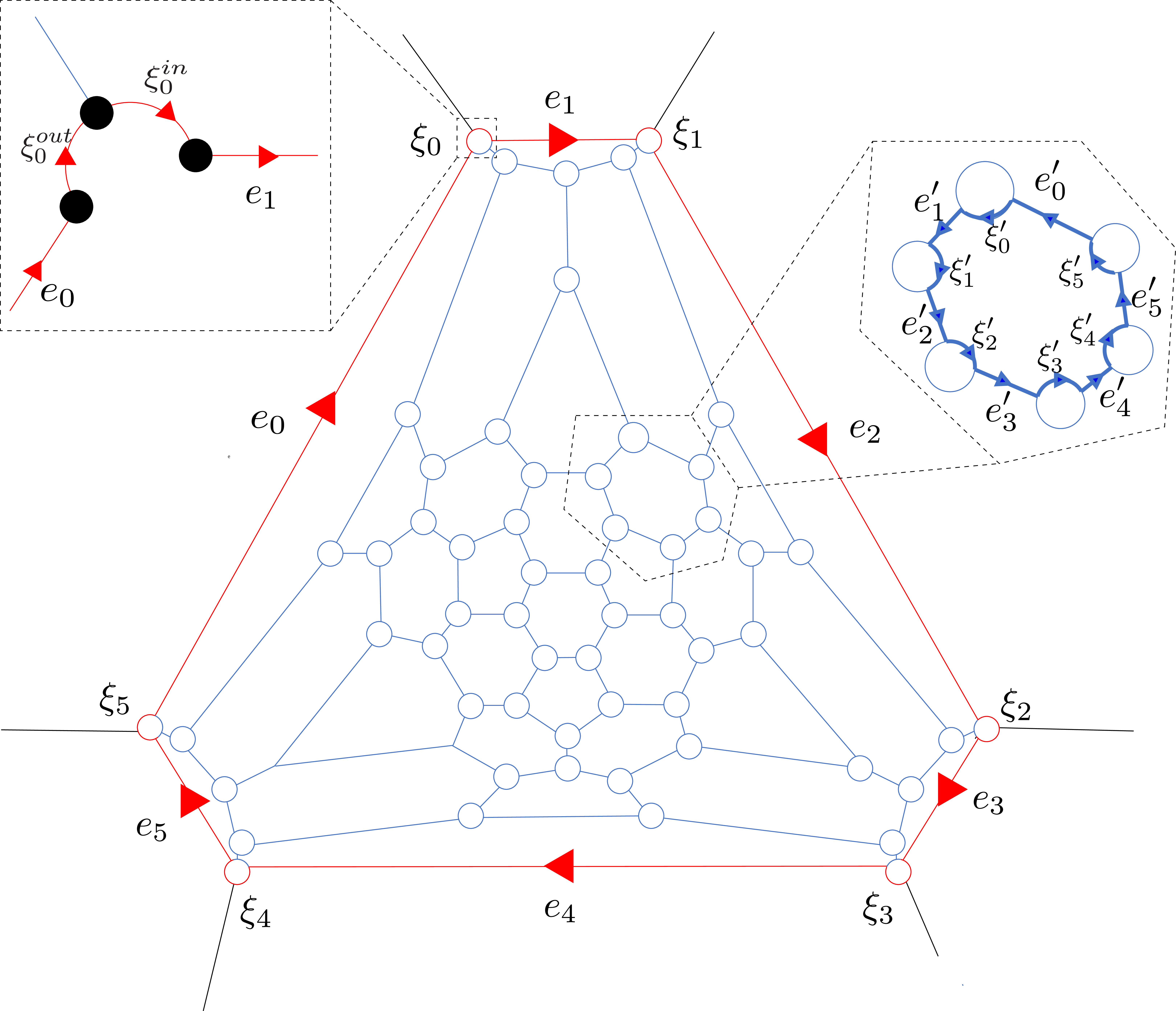}
    \caption{Internal and external facial walks on the truncated icosahedron: The truncated icosahedron $G$ is depicted as the planar graph. The rotation is assigned clockwise at each vertex. This graph is the blow up graph $\tilde{G}$: the circles colored by white are islands whose orientation is clockwise. The closed oriented cycle around the boundary of this graph  $f_*:=(e_0,\xi_0,\dots,e_5,\xi_5)$ is the external facial walk, where $\xi_j=(\xi_j^{out},\xi_j^{in})$ ($j=0,\dots,5$). The hexagons and  pentagons in the internal correspond to the internal facial walks. For example, the hexagon surrounded by the dotted line in the figure depicts the internal facial walk $f=(e_0',\xi_0',\dots,e_5',\xi_5')$. 
    The external facial function of $f_*$ on $\supp(f_*)$ is expressed by 
    $\gamma_{f_*}^{ex}(\xi_0^{in})=b\omega \eta_0$, 
    $\gamma_{f_*}^{ex}(e_1)=\omega \eta_0$,
    $\gamma_{f_*}^{ex}(\bar{e}_1)=db\omega \eta_0$,
    $\gamma_{f_*}^{ex}(\xi_1^{out})=b\omega\eta_0$,
    $\cdots$, 
    $\gamma_{f_*}^{ex}(\xi_5^{in})=b\omega \eta_5$, 
    $\gamma_{f_*}^{ex}(e_0)=\omega \eta_5$,
    $\gamma_{f_*}^{ex}(\bar{e}_0)=db\omega \eta_5$,
    $\gamma_{f_*}^{ex}(\xi_0^{out})=b\omega\eta_5$.
The internal facial function of $f$ on $\supp(f)$ is, for example, described by 
    $\gamma_f^{in}(e_j')=\omega^j$, $\gamma_f^{in}(\bar{e}_j')=d\omega^j$, $\gamma_f^{in}(\xi_j)=b\omega^j$, $(j=0,\dots,5)$. 
    }
    \label{fig:soccer1}
\end{figure}
\begin{figure}[hbtp]
    \centering
    \includegraphics[keepaspectratio, width=150mm]{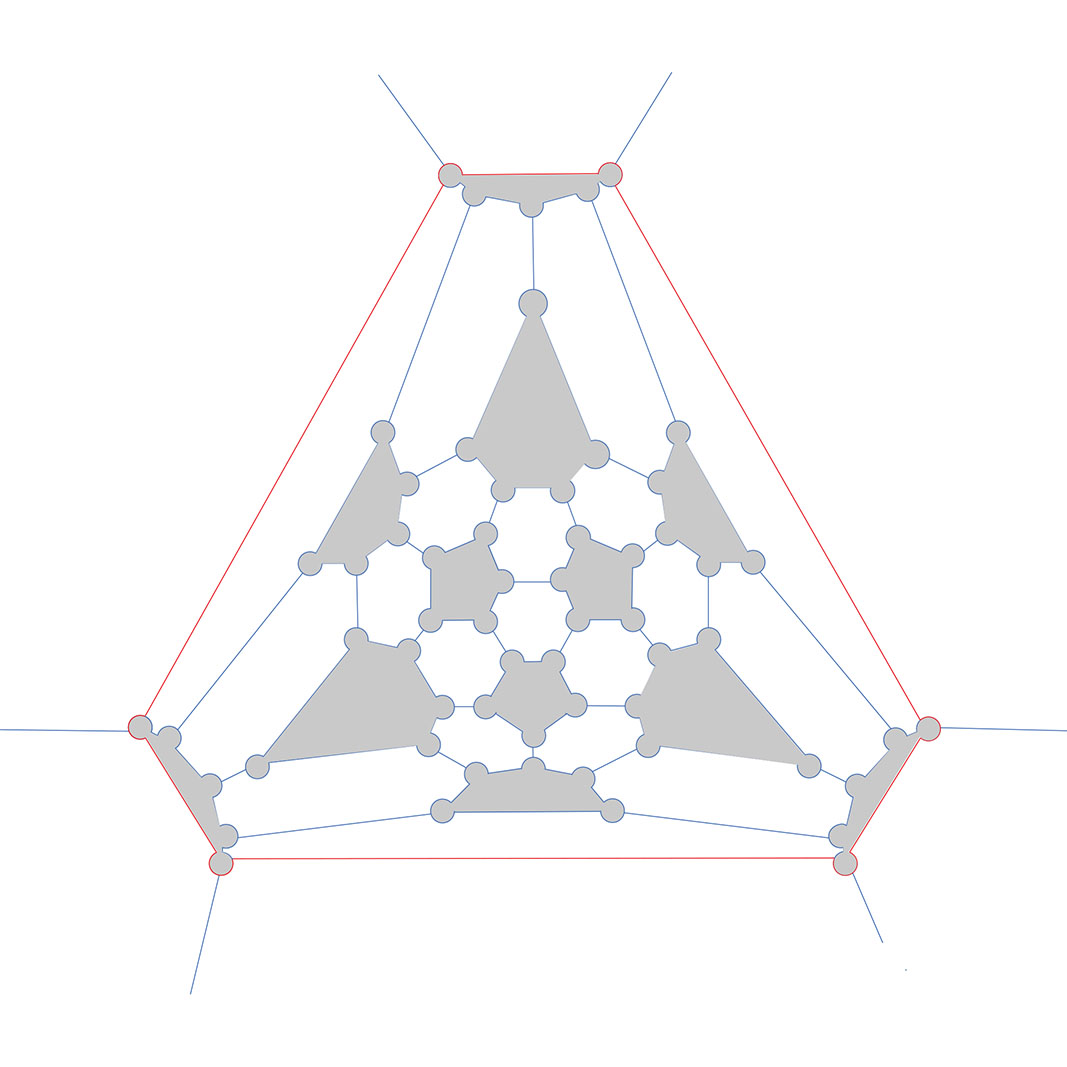}
    \caption{The luminous faces:  When $\omega^6=1$ and $\omega^5\neq 1$, the luminous faces for the stationary state are colored by white in this figure. }
    \label{fig:soccer2}
\end{figure}
\subsection{Stationary state for $\omega=1$}
Now the rest of our task for describing the stationary state is to characterize the coefficients $c_f$'s in Proposition~\ref{prop:support}. 
Let us obtain them as follows. 
By Proposition~\ref{prop:support}, if there are no internal faces such that $\omega^{|f|_G}=1$, then $\Psi_\infty=\sum_{g\in F^{ex}}\tilde{\gamma}_f^{ex}$. 
In this section, we are interested in the case, where $\omega=1$ because a quantum walk penetrates every internal faces from Proposition~\ref{prop:support}. 
\subsubsection{Preparation}
Let us set the tails satisfying  that there is only one external facial walk, namely $f_*$, and each tail joins to each vertex of $f_*$. 
In the following, let us determine the coefficients $\{c_f\}_{f\in F^{in}}$ in this setting. 

The set of the facial walks on $\tilde{G}$ are denoted by $F=F^{in}\cup F^{ex}$ with $|F|=p$ and  we set the tails so that the external facial walk is $F^{ex}=\{f_*\}$ and each  tail joins to each vertex of $f_*$. 
For an internal facial walk $f=(e_0,\xi_0,\dots,e_{s-1},\xi_{s-1})\in F^{in}$, we define $[f]:=\{e_0,\dots,e_{s-1}\}$ and $[\bar{f}]:=\{\bar{e}_0,\dots,\bar{e}_{s-1}\}$. 
We set the length of $f\in F^{in}$, $s$, by $|f|_G$. 
For internal facial walks, $f$ and $f'$,  we set 
\[ m_{f,f'}:=\# ([f]\cap [\bar{f'}]); \]
in particular, we write $m_f$ for $m_{f,f}$. 
We say that the faces $f$ and $g$ are adjacent each other if $m_{f,g}>0$; the integer value $m_{f,g}$ is called a multiplicity.  
The internal facial function $\gamma^{in}_f$ defined in Definition~\ref{def:facialfunction} is rewritten by 
\begin{equation}\label{eq:innerproduct}
\gamma_f^{in}(e) = 
\begin{cases}
1+\bs{1}_{\{\bar{e}_j\in[f]\}}(e_j)\;d & \text{: $e=e_j$ $(j=0,1,\dots,s-1)$,}\\
\bs{1}_{\{\bar{e}_j\notin[f]\}}(e_j)\; d &\text{: $e=\bar{e}_j$ $(j=0,1,\dots,s-1)$,}\\
b & \text{: $e=\xi_j$ $(j=0,1,\dots,s-1)$,} \\
0 & \text{: otherwise.}
\end{cases}
\end{equation}
For the set of facial walks $F=F^{in} \cup \{f_*\}$ with $F^{ex}=\{f_*\}$,  $|F|=p$, we set $F^{in}=\{f_1,\dots,f_{p-1}\}$ and 
consider the Gram matrix of $\{\gamma^{in}_{f}\}_{f\in F^{in}}$ defined by 
\[ (\mathcal{M})_{i,j} = \langle \gamma_{i},\gamma_{j}  \rangle \]
for any $i,j\in \{1,\dots,p-1\}$, where $\gamma_i=\gamma_{f_i}$ $(i=1,\dots,p-1)$. 
Using the expression for (\ref{eq:innerproduct}), let us see 
\[ (\mathcal{M})_{i,j} := 2d\;m_{i,j}+2|f_i|_G\;\delta_{i,j},  \]
where $m_{i,j}:=m_{f_{i},f_{j}}$. 
If $i\neq j$, 
\[ \langle \gamma_i,\gamma_j \rangle=m_{i,j}\left\langle \begin{bmatrix}1 \\d \end{bmatrix},\;\begin{bmatrix}d \\1 \end{bmatrix}\right\rangle=m_{i,j} \times 2d. \]
On the other hand, if $i=j$, then putting $s=|f_j|$
\begin{align*}
||\gamma_i||^2 &= |b|^2 s 
+ \sum_{k=0}^{s-1}\left( 1+\bs{1}_{\bar{e}_k\in[f_j]}(e_k) d\right)^2 
+ \sum_{k=0}^{s-1}\left( d\;\bs{1}_{\bar{e}_k\notin[f_j]}(e_k) \right)^2 \\
&= |b|^2 s + (s+2d\;m_i+d^2m_i)+d^2(s-m_i) \\
&= 2|f_i|+2d \;m_i.
\end{align*}

To determine the coefficients $c_i:=c_{f_i}$  $(i=1,\dots,p-1)$ in (\ref{eq:st}), let us set the vector $\bs{c}:=[\;c_1,\dots,c_{p-1}\;]^\top$, $\bs{b}:=[\;\langle \gamma_1,\psi^{ex} \rangle,\dots,\langle \gamma_{p-1},\psi^{ex} \rangle\;]^\top$, where $\psi^{ex}:=\sum_{g\in F^{ex}}\gamma_g^{ex}=\gamma_{f_*}^{ex}$. Then by taking the inner product of $\gamma_{i}$ to both sides of (\ref{eq:st}), we have 
\begin{equation}\label{eq:c} 
\bs{c}=\mathcal{M}^{-1} \bs{b}, \end{equation}
because $\langle \psi_\infty,\gamma_f^{in}\rangle=0$ for any $f\in F^{in}$. 
Note that $\gamma_j$'s are linearly independent and $\mathcal{M}$ is its Gram matrix, then $\mathcal{M}^{-1}$ exists. 
\subsubsection{Geometric representation}
Let $G^*$ be the dual graph of $G$ induced by the rotation $\rho$. Note that the dual graph $G^*$ depends on the rotation $\rho$, equivalently, the underlying embedding. 
Naturally, $G^*$ is a usual dual graph of $G$ which embedded into the suitable orientable surface with respect to $\rho$. 
It holds that for any $f_i\in V(G^*)$, 
\[ \deg_{G^*}(f_i)=|f_i|_G=\sum_{j}m_{i,j}. \]
From now on, we put together every multi-edge of $G^*$ as a simple edge and remove every self-loop from $G^*$ and add one self-loop again to every vertex of $G^*$ except the sink vertex $f_*$; such a resulting graph is denoted by $\SL{G^*}$.  
Then $\SL{G^*}$ is described by  
\begin{align*}
    V(\SL{G^*}) &= F^{in}(G)\cup F^{ex}(G); \\
    \{f,  g\}\in E(\SL{G^*}) & \Leftrightarrow ``m_{f,g}\neq 0 \text{ if $f\neq g$}" \text{ or  ``$f\neq f_*$ if $f=g$"}.
\end{align*}
The self-loop which was newly added to a vertex in $V(G^*)\setminus\{f_*\}$ is called the potential self-loop.  
The set of all potential self-loops in $\SL{G^*}$ is denoted by $S(\SL{G^*})$. 
In this paper, we regard the self-loop as an odd cycle and the isolated vertex as a tree.  
If a subgraph $H\subset \SL{G^*}$ has a unique potential self-loop, say $s$, and $H\setminus\{s\}$ is a tree, then $H$ is called a potential tree. 

Now we introduce the following two important families of the spanning subgraphs of $\SL{G^*}$ whose connected components are constructed by trees and 
potential trees. 

\begin{definition}[The families of spanning subgraphs for describing the stationary state]
\begin{align*}
    \mathcal{H}(G;f_*) &:=\{ \text{$H$ is a spanning subgraph of $\SL{G^*}$ } \;|\; \\ 
    & \qquad\qquad\qquad\text{(i) the connected component of $H$ including $f_*$ is a tree;} \\
    & \qquad\qquad\qquad\text{(ii) the other connected components of $H$ are potential trees. }
    \}\\
    \mathcal{H}(G;f_*;f,g) &:=\{ 
    \text{$H$ is a spanning subgraphs of $\SL{G^*}$ } \;|\; \\ 
    & \qquad\qquad\qquad
    \text{(i)' there are exactly $2$ trees in the connected components of $H$, and } \\
    & \qquad\qquad\qquad\qquad\qquad
    \text{one includes $f_*$,}\\
    & \qquad\qquad\qquad\qquad\qquad
    \text{the other includes both $f$ and $g$;} \\
    & \qquad\qquad\qquad\text{(ii) the other connected components of $H$ are potential trees. }
    \}
\end{align*} 
for any $f,g\in F^{in}$. 
\end{definition}

The vertex set  $V(\stackrel{\circ}{G^*})$ is denoted by $\{u_1,\dots,u_{p-1},u_p\}$, where $u_i$ corresponds to $f_i\in F^{in}$ for $i=1,\dots,p-1$
 and $u_p$ corresponds to $f_*\in F^{ex}$. 
For an edge $e\in E(\SL{G^*})$ with its end vertices $u_i$ and $u_j$,  we set $m_e:= m_{i,j}$. 
The weight on $E(\stackrel{\circ}{G^*})$ is denoted by 
\begin{equation}\label{eq:weight}
w(e)=\begin{cases} -2d \; m_e & \text{: $e\in E(\stackrel{\circ}{G^*})\setminus S(\stackrel{\circ}{G^*})$,} \\
2(|f_i|_G+d\;\deg_{G^*} (o(e)))=2(d+1)\deg_{G^*}(u_i)  & \text{: $e\in S(\stackrel{\circ}{G^*})$,}
\end{cases}
\end{equation}
where $|f|_G=|f|$ is the length of the facial walk $f$ following the rotation $\rho$ in $G$, and $\deg_{G^*}(f):=\sum_{g}m_{f,g}$ is the degree in $G^*$ for any $f\in V(G^*)$. 
The following quantities induced by the above families of spanning subgraphs are important values to describe the stationary state. 
\begin{definition}[Weights of spanning subgraph and two families of spanning subraphs] 
\noindent \\
For $H\subset \stackrel{\circ}{G^*}$, the weight of $H$ is defined by  $\mathcal{H}(G;f_*;f,g)$: 
\[ W(H)=\prod_{e\in H}w(e),  \]
For the families of $\mathcal{H}(G;f_*)$ and $\mathcal{H}(G;f_*;f,g)$, the weights of $\mathcal{H}(G;f_*)$ and $\mathcal{H}(G;f_*;f,g)$ are defined by 
\begin{align*}
    \iota_1(G;f_*) &= \sum_{H\in \mathcal{H}(G;f_*)}W(H), \\
    \iota_2(G;f_*;f,g) &= \sum_{H\in \mathcal{H}(G;f_*;f,g)}W(H).
\end{align*} 
\end{definition}

Now we are ready to state the theorem for the stationary state. 
\begin{theorem}\label{thm:stationary}
Assume $\omega=1$, $d\in\mathbb{R}$ and $G$ has $p-1$ internal faces $f_1,\dots,f_{p-1}$ and the external face $f_*$. 
The internal facial functions in Definition~\ref{def:facialfunction} are denoted by $\gamma_j$'s $(j=1,\dots,p-1)$ and the external facial function in Definition~\ref{def:facialfunction} is denoted by $\gamma^{ex}_*$. 
Then the stationary state is  
\begin{equation}\label{c}
\psi_{\infty}=\left(1-\sum_{\ell,m=1}^{p-1}\frac{\iota_2(G;f_*;f_\ell,f_m)}{\iota_1(G;f_*)}\Gamma_{\ell,m}^{in}\right)\gamma^{ex}_*, 
\end{equation}
where $\Gamma_{\ell,m}^{in}:=\gamma_{\ell}^{in}\;{\gamma_{m}^{in}}^*$ $(m=1,\dots,p-1)$. 
\end{theorem}
\begin{proof}
Let $G^*=(V^*,E^*)$ be the dual graph of the original graph $(G;\delta V;\rho)$. 
Let us set $\vec{E^*}$ as the set of oriented arcs of $\SL{G^*}$ such that if $e\in \vec{E^*}$, then $\bar{e}\notin \vec{E^*}$ and $ \{|e| \;|\; a\in \vec{E} \}=E(\SL{G^*})$; we fix one direction in the $2^{|E(\SL{G^*})|}$ choices. 
Let $B\in  \mathbb{C}^{(V^*\setminus\{u_p\}) \times \vec{E^*}}$ be one oriented incidence matrix such that
\[ B[u,a]=\begin{cases} 1 & \text{: $a\in \vec{E^*}\setminus S$, $t(a)=u$,} \\
-1 & \text{: $a\in \vec{E^*}\setminus S$, $o(a)=u$,} \\
1 & \text{: $a\in S$, $t(a)=o(a)=u$, }\\
0 & \text{: otherwise}
\end{cases} \]
for any $u\in V$ and $a\in \vec{E^*}$. 
Let $D_w$ be the diagonal matrix of $\mathbb{C}^{(\vec{E^*}\cup S)\times (\vec{E^*}\cup S)}$ such that 
\[ D_w[a,b]=\delta_{a,b}\; w(a), \]
for any $a,b\in \vec{E^*}\cup S$. 
Then we can reexpress the Gram matrix $\mathcal{M}$ by 
\begin{align*}
\mathcal{M} &= BD_wB^*, 
\end{align*}
where, $B^*\in \mathbb{C}^{(\vec{E^*}\cup S)\times (V^*\setminus\{u_p\}) }$ is the adjoint of $B$. 
Then $\mathcal{M}$ can be also expressed by 
\[ \mathcal{M}=L+\mathcal{V}', \]
where $L\in \mathbb{C}^({V^*\setminus\{f_p\})\times (V^*\setminus\{f_p\})}$  is the weighted Laplacian such that
\[ L[f,g]=2d\times \begin{cases} -m_{f,g} & \text{: $f\neq g$, }\\
 \sum_{h\neq f}m_{f,h} & \text{: $f=g$ }
\end{cases} \]
and $\mathcal{V}'$ is the diagonal matrix described by 
\[ \mathcal{V}'[f,g]=\begin{cases}
2d\;m_f+2\sum_{h\in V^*\setminus\{f_p\}}m_{f,h} & \text{: $f=g$,}\\ 
0 & \text{: otherwise.}
\end{cases}\]
By applying Theorem~3.1 in \cite{HS_Resolvent} to $\mathcal{M}$, we have 
\begin{equation}\label{eq:inverse}
\mathcal{M}^{-1}[f,g]=\frac{\iota_2(G;f_*.f,g)}{\iota_1(G;f_*)}.
\end{equation}
Then  inserting (\ref{eq:c}) into 
(\ref{eq:st}) and using the expression for $\mathcal{M}^{-1}$ in (\ref{eq:inverse}), we have 
\begin{align*}
\psi_\infty &= 
\gamma_p^{ex}-\sum_{f\in F^{in}} \bs{c}[f] \;\gamma_f^{in} & (\text{by } (\ref{eq:c}))\\
&= \gamma_p^{ex}-\sum_{f\in F^{in}} (\mathcal{M}^{-1}\bs{b})[f] \gamma_f^{in} & (\text{by } (\ref{eq:st})) \\
&= \gamma_p^{ex}-\sum_{f,g\in F^{in}} \frac{\iota_2(\Gamma;f,g)}{\iota_1(\Gamma;f_*)} \Gamma_{f,g}^{in} \gamma_*^{ex}, &(\text{by } (\ref{eq:inverse})) 
\end{align*}
which implies the desired conclusion. 
\end{proof}

\subsubsection{Example: computation of the stationary state on the tetrahedron}\label{sect:4.3.3}
Let us concretely compute the $c_f$'s in the case for the tetrahedron as the original graph with the vertex set $\{0,1,2,3\}$. The boundary set is $\{0,1,2\}$. Each rotation is assigned to each vertex clockwise.  
The resulting rotation tailed graph is $(G;\delta V;\rho)$. 
See Figure~\ref{fig:2}; the graph $(G;\delta V; \rho)$ and the dual graph $G^*$ are depicted. The graph $\SL{G^*}$ has self-loops on the black vertices in $G^*$.
Then in this case, the weight of the each arc of $E(\stackrel{\circ}{G^*})$ in (\ref{eq:weight}) is reduced to 
\[ w(e)=\begin{cases} -2d & \text{: $e\in E(\stackrel{\circ}{G^*})\setminus S(\SL{G^*})$,}\\
2(1+d)\deg_{G^*}(o(e))=6(1+d) & \text{: $e\in S(\SL{G^*})$.}
\end{cases}\]
Let us set $q:=6(1+d)$ and $p:=-2d$. 
The list of all spanning forests of the dual graph of $G^*$ is depicted in Figure~\ref{fig:3}.   
Figure~\ref{fig:4} depicts the list of the spanning forests of the dual graph of $G^*$ which induce the spanning subgraph in $\mathcal{H}(G;f_*;f_{i}^{in},f_{j}^{in})$. 
Each forest induces some spanning subgraph of $\mathcal{H}(G;f_*)$ and $\mathcal{H}(G;f_*;f_{i}^{in},f_{j}^{in})$.  In Figures~\ref{fig:3} and \ref{fig:4}, ``GRAPH$^{\times m}$" means that this graph induces $m$ kinds of potential trees in $\mathcal{H}(G;f_*)$ and $\mathcal{H}(G;f_*;f_i^{in},f_j^{in})$, respectively.
We omit ``$\times m$" when $m=1$. 

We divide each set of the spanning forests into the equivalent class following the relation: 
for any subgraph $H$ and $H'$,  $H\stackrel{W}{\sim} H'$ iff $W(H)=W(H')$. 
For example, for any spanning forest in the class ``$p^2q$",  the induced spanning subgraph gives the same weight $p^2q$ in Figure~\ref{fig:3}; the number of such induced subgraphs is $3\times 3+1\times 9+2\times 3=24$.   
Let us denote the number of connected component of the spanning forest $H$, as $\omega (H)$. Note that $\omega (H)$ determines the weight of $H$; 
\[W(H)=p^{|E(H)|-(\omega(H)-1)}q^{\omega(H)-1}. \]
We can compute $\iota_1$ and $\iota_2$ by using Figures~\ref{fig:3} and \ref{fig:4} as follows: 
\begin{align*}
\iota_1 &= 16p^3+24p^2q+9pq^2+q^3 = 8(d-3)^2(3+2d), \\
\iota_2(\ell,m) &= 
\begin{cases}
8p^2+6pq+q^2 = 4(9-d^2) & \text{: $\ell=m$,} \\
4p^2+pq=4d(d-3)  & \text{: $\ell\neq m$.}
\end{cases}
\end{align*}
Then 
\[ \mathcal{M}^{-1}[\ell,m]=\frac{\iota_2(\ell,m)}{\iota_1}=\frac{1}{2(d-3)(2d+3)}\begin{cases}-(d+3) & \text{: $\ell=m$,}\\ d & \text{: $\ell\neq m$.}\end{cases} \]
The scattering matrix is described by 
\[ S=\frac{\omega bc}{1-(a\omega)^3}\begin{bmatrix} (a\omega)^2 & (a\omega) & 1 \\
1 & (a\omega)^2 & (a\omega) \\
(a\omega) & 1 & (a\omega)^2 
\end{bmatrix}+dI_3 \]
for any $\omega$ with $|\omega|=1$. 
If the internal flow is set by $\bs{\alpha}_{in}=[1,1,1]^\top$, then the outflow $\bs{\beta}_{out}$ is 
\[ \bs{\beta}_{out}[j]=\frac{\omega bc}{1-a\omega}+d, \]
which implies 
\[ \eta_j=\frac{\omega^{-1}}{bc}\left(\bs{\beta}_{out}[j+1]-d\bs{\alpha}_{in}[j+1]\right)\eta=\frac{1}{1-a\omega} \]
for $j=0,1,2$. 
The external facial function for the internal facial walk $(e_0,\xi_0,e_1,\xi_1,e_2,\xi_2)$, where $\xi_j=(\xi_j^{out},\xi_j^{in})$ ($j=0,1,2$), can be expressed by 
\begin{align}
\gamma_{f_*}^{ex}(\xi_0^{in}) &= b\eta,\; 
\gamma_{f_*}^{ex}(e_1) = \omega \eta, \;
\gamma_{f_*}^{ex}(\bar{e}_1) = db\omega \eta,\; 
\gamma_{f_*}^{ex}(\xi_1^{out}) = b\omega \eta, \notag\\
\gamma_{f_*}^{ex}(\xi_2^{in}) &= b\eta,\; 
\gamma_{f_*}^{ex}(e_2) = \omega \eta, \;
\gamma_{f_*}^{ex}(\bar{e}_2) = db\omega \eta,\; 
\gamma_{f_*}^{ex}(\xi_2^{out}) = b\omega \eta, \notag \\
\gamma_{f_*}^{ex}(\xi_2^{in}) &= b\eta, \;
\gamma_{f_*}^{ex}(e_0) = \omega \eta, \;
\gamma_{f_*}^{ex}(\bar{e}_0) = db\omega \eta,\; 
\gamma_{f_*}^{ex}(\xi_0^{out}) = b\omega \eta. \label{eq:external0}
\end{align}
On the other hand, the internal facial function of the walk $f_j=(e_0',\xi_0',e_1',\xi_1',e_2',\xi_2')$, where $e_0'= \bar{e}_j$, is expressed by 
\begin{align}
\gamma^{in}_f(e_0') &= 1,\;
\gamma^{in}_f(\bar{e}_0') = d,\;
\gamma^{in}_f(\xi_0') = b, \notag \\
\gamma^{in}_f(e_1') &= \omega,\;
\gamma^{in}_f(\bar{e}_1') = d \omega,\;
\gamma^{in}_f(\xi_1') = b \omega, \notag\\
\gamma^{in}_f(e_2') &= \omega^2,\;
\gamma^{in}_f(\bar{e}_2') = d \omega^2,\;
\gamma^{in}_f(\xi_2') = b \omega^2 \label{eq:internal0}
\end{align}
for any $|\omega|=1$. 
Then if $\omega=1$, 
\begin{align*}
\bs{b}=[\;\langle \gamma_0^{in},\gamma^{ex}_{f_*} \rangle,\;\langle \gamma_1^{in},\gamma^{ex}_{f_*} \rangle,\;\langle \gamma_2^{in},\gamma^{ex}_{f_*} \rangle \;]^\top 
= d(b+1)\;\eta\;[\; 1,1,1 \;]^\top. 
\end{align*}
Since $\bs{c}=\mathcal{M}^{-1}\bs{b}$ for $\omega=1$, we obtain the coefficients $c_j$ $j=0,1,2$ as follows: 
\begin{equation}\label{eq:cjs}
c_j=\bs{c}[j]=\nu:=\sum_{i=0}^3\mathcal{M}[j,i]\bs{b}[i]=d(b+1) \eta \frac{1}{2(2d+3)}=
\frac{d (b+1) }{1-a}  \frac{1}{2(2d+3)}.
\end{equation}
Since $c_{f}=\sum_{i}\frac{\iota_2(G;f_*;f,g)}{\iota_1(G;f_*)}\langle \gamma_g^{in},\psi^{ex}\rangle$ for any $f\in F^{in}$, 
then inserting (\ref{eq:external0}), (\ref{eq:internal0}) and (\ref{eq:cjs}) into (\ref{eq:c}) in  Theorem~\ref{thm:stationary}, we obtain the stationary state as follows. 
Let the external facial walk be $(\xi_0^{in},e_0,\xi_1^{out},\xi_1^{in},e_1,\xi_2^{out},\xi_2^{in},e_2,\xi_0^{out})$, where $(\xi_k^{out},\xi_k^{in})\in \delta A_{qy}$ and $e_k\in A_{br}$. 
Labeling the arcs of each internal facial walk $(\xi_{0,k},e_{0,k},\xi_{1,k},e_{1,k},\xi_{2,k},e_{2,k})$, where $\xi_{j,k}\in A_{is}$ and $e_{j,k}\in A_{br}$ $(j=0,1,2)$, satisfying that $e_k=\bar{e}_{0,k}$ $(k\in \{0,1,2\})$, we have 
\begin{align*}
\phi_\infty(\xi_k^{in})&= b\eta,\;
\phi_\infty(e_k)= \phi_\infty(\bar{e}_{0,k})=\eta -\nu d,\; 
\phi_\infty(\xi_{k+1}^{out}) = b \eta, \\
\\
\phi_\infty(\xi_{0,k})&= -\nu b,\; 
\phi_\infty({e}_{0,k})=d \eta -\nu,\;  
\phi_\infty(\xi_{1,k})= -\nu b \\
\phi_\infty(e_{1,k})&= -\nu (1+d),\; 
\phi_\infty(\xi_{2,k}) = -\nu b,\; 
\phi_\infty(e_{2,k}) = -\nu (1+d)
\end{align*}
for $k=0,1,2$.
See Figure~\ref{fig:closedsurface}.

\begin{figure}[hbtp]
    \centering
    \includegraphics[keepaspectratio, width=130mm]{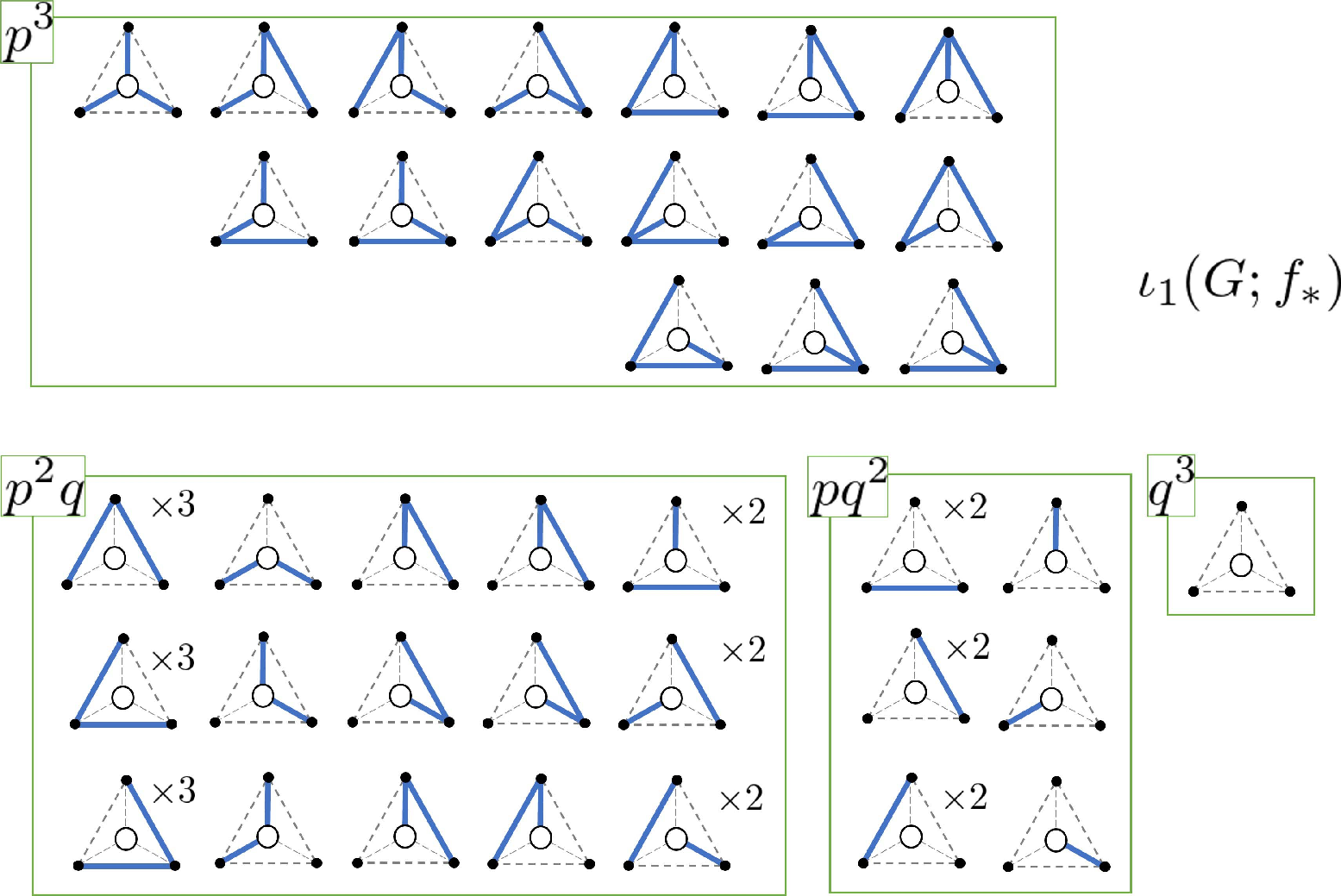}
    \caption{The list of the spanning forests of $G^*$: 
    The white vertex depicts $u_p$ corresponding to $f_*\in F^{ex}$. The weight on the edge is $p$ while the weight on the potential self loop is $q$.  
    ``GRAPH"$^{\times m}$ means that ``GRAPH" has $m$ possibilities to become a potential tree by adding  a potential self loop. For example, the graph of the upper left corner for ``$p^2q$", this graph induces the $3$ kinds of  potential trees by adding the self loop to each vertex except $f_*$. Then the coefficients of $p^3$, $p^2q$, $pq^2$ and $q^3$ for $\iota_1(G;f_*)$ are 
    $1\times 16$, 
    $3\times 3+2\times 3+1\times 9=24$, 
    $2\times 3+1\times 3=9$ 
    and 
    $1$, respectively. 
    }
    \label{fig:3}
\end{figure}
\begin{figure}[hbtp]
    \centering
    \includegraphics[keepaspectratio, width=150mm]{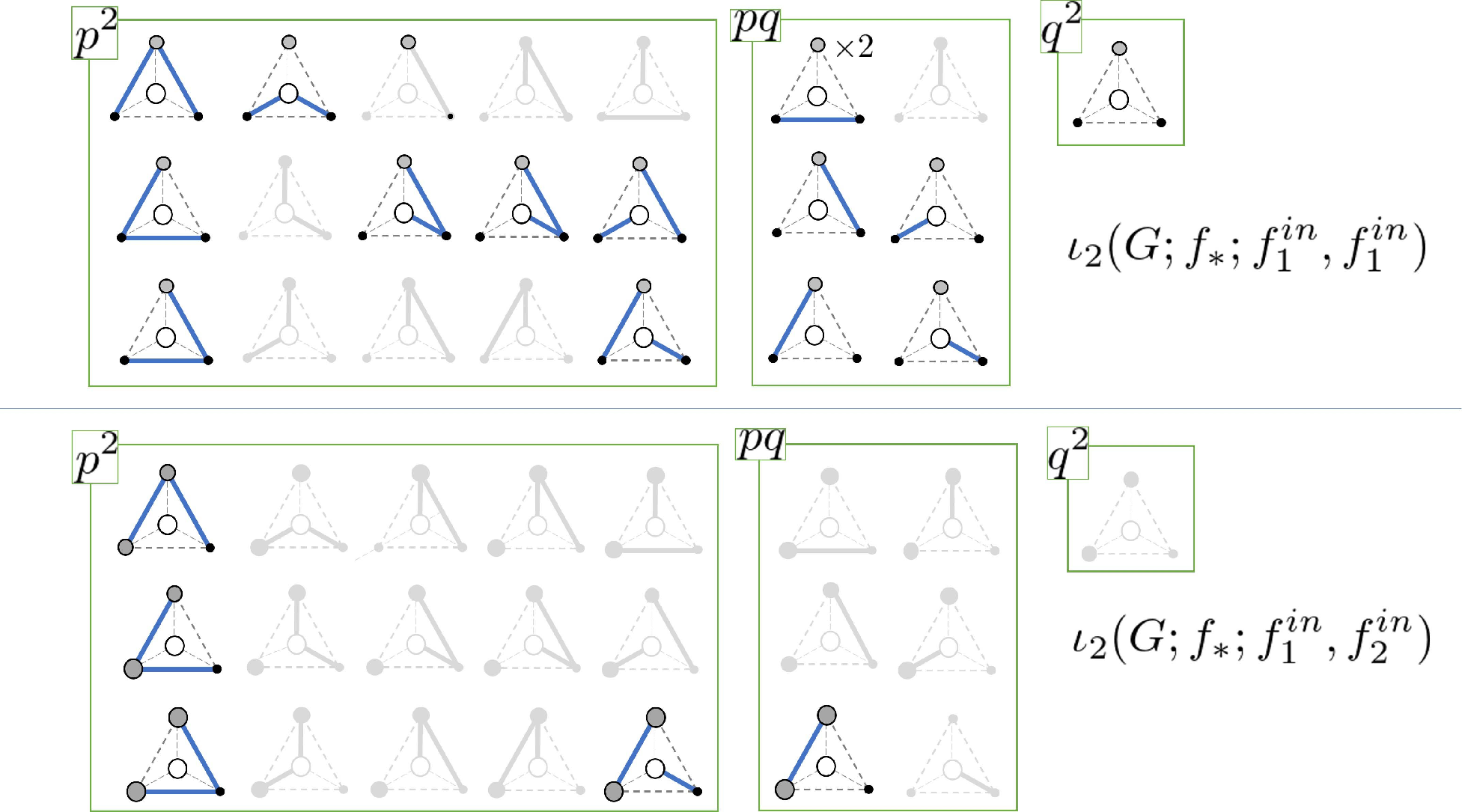}
    \caption{The list of the spanning forests of $G^*$ which induces $\mathcal{H}(G;\delta V;f_i^{in},f_j^{in})$}
    \label{fig:4}
\end{figure}

\begin{figure}[hbtp]    \centering    \includegraphics[keepaspectratio, width=130mm]{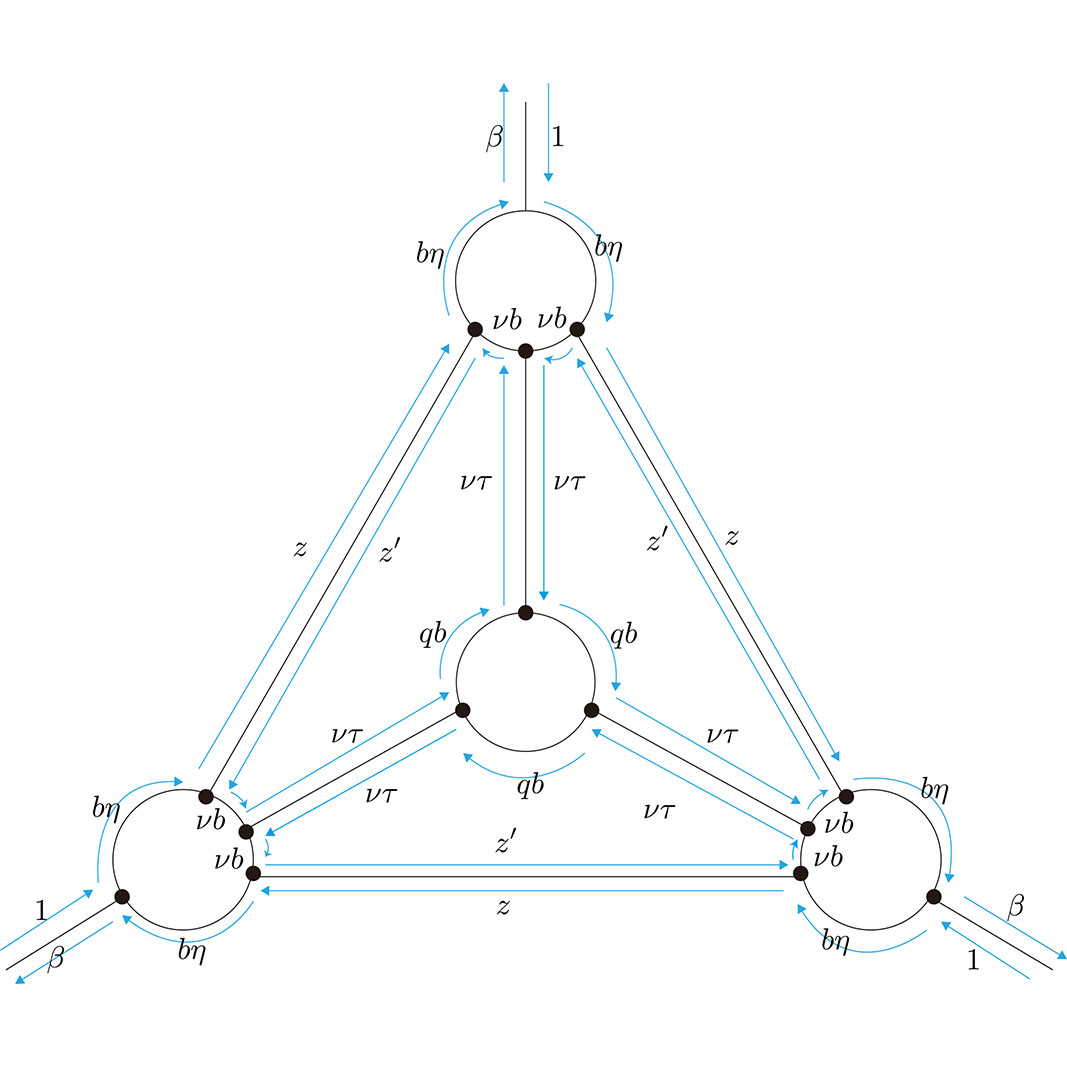}    \caption{The stationary state on the tetrahedron: Each value on each arc represents the amplitude of the stationary state at this arc, where $z=\eta-\nu d$, $z'=d\eta-\nu d$, $\tau=1+d$.}    \label{fig:closedsurface}\end{figure}

\noindent\\
\noindent {\bf Acknowledgments}
Yu.H. acknowledges financial supports from the Grant-in-Aid of
Scientific Research (C) Japan Society for the Promotion of Science (Grant No.~18K03401, No.~23K03203). 
E.S. acknowledges financial supports from the Grant-in-Aid of
Scientific Research (C) Japan Society for the Promotion of Science (Grant No.~19K03616) and Research Origin for Dressed Photon.



\begin{small}
\bibliographystyle{jplain}

\end{small}

\end{document}